\documentclass{article}
\usepackage[colorlinks]{hyperref}
\usepackage{amsmath,amsthm,amssymb,cleveref,mathtools,graphicx}
\usepackage{unicode-math-input}
\usepackage{xcolor}
\usepackage{import}
\usepackage{algorithm}
\usepackage{algpseudocode}
\crefname{claim}{claim}{claims}
\DeclarePairedDelimiter{\abs}{|}{|}
\DeclarePairedDelimiter{\Abs}{\Vert}{\Vert}
\algnewcommand\Print{\textbf{print} }%
\algrenewcommand\textproc[1]{\textnormal{\textsc{#1}}}
\DeclareMathOperator\id{id}
\newtheorem{definition}    {Definition}
\newtheorem{remark}        {Remark}
\newtheorem{proposition}   {Proposition}
\newtheorem{theorem}       {Theorem}
\newtheorem{lemma}         {Lemma}
\newtheorem{example}       {Example}
\begin{document}
\umiDeclareMathChar{˜}{\tilde}
\umiDeclareMathChar{‾}{\overline}
\umiDeclareMathChar{ˆ}{\widehat}
\umiDeclareMathChar{˘}{\breve}
\umiDeclareMathChar{Õ}{\widetilde O}
\umiDeclareMathChar{‥}{\mathbin{.\,.}}  
\newcommand\clamp{\operatorname{clamp}}
\newcommand\RelCond{\operatorname{RelCond}}
\newcommand\RelNewtonCond{\operatorname{RelNewtonCond}}
\newcommand\FFT{\operatorname{FFT}}
\newcommand\argmax{\operatorname{argmax}}
\newcommand\minmax{\operatorname{minmax}}
\newcommand\round{\operatorname{round}}
\newcommand\Ball{\operatorname{Ball}}
\newcommand\Int{\operatorname{Int}}
\newcommand\val{\operatorname{val}}
\newcommand\var{\operatorname{var}}
\newcommand\polylog{\operatorname{polylog}}
\newcommand\poly{\operatorname{poly}}
\newcommand\RescaleRound{\Call{RescaleRound}}
\newcommand\RescaledConvolution{\Call{RescaledConvolution}}
\newcommand\AutoRescaledConvolution{\Call{AutoRescaledConvolution}}
\newcommand\Unscale{\Call{Unscale}}
\newcommand\IntegerConvolution{\Call{IntegerConvolution}}
\newcommand\ImportFigure[2]{
	\begin{figure}
		\centering
		\import{figures}{#1.pdf_tex}
		\caption{#2}
		\label{#1}
	\end{figure}
}
\pdfglyphtounicode{Delta}{0394}  

\title{Faster Stable Numerical Polynomial Multiplication}
\author{Hong Duc Bui}
\maketitle

\begin{abstract}
	In the preprint \cite{stablemult}, van der Hoeven considers the problem of multiplying two
	polynomials with floating-point coefficients, and give an algorithm to compute the product
	with small relative Newton error in time $O(np \log(np))$, where $n$ is the degree
	and $p$ is the required precision. In this paper, we describe a significantly simpler
	algorithm with the same time complexity and error bound.

	Independently, Bringmann and Cassis considered the near-convex min-plus convolution problem
	in \cite{bringmann_2023}, and presented an algorithm to solve that problem.
	We observe that our algorithm can be adapted to that problem
	to speed up Bringmann and Cassis' algorithm by a logarithmic factor.
\end{abstract}

\section{Introduction}

Multiplying two polynomials is a well-studied problem in both theoretical and practical computer science.
In this article, we consider the problem of multiplying two polynomials with coefficients given
as floating-point numbers.

We work in the word RAM model, 
where memory is divided into words,
each word has $w ≥ 1$ bits,
and the standard arithmetic operations (addition, subtraction, multiplication, division,
bit shifts) can be done in $O(1)$ time.
This model is designed so that the time complexities of algorithms reflect running times on real computers. 

In this model, a positive integer can be represented in radix $2^w$,
such that each base-$2^w$ digit fits in a machine word.\footnote{
	The authors of GMPlib calls base-$2^w$ digits \emph{limbs}, 
	we would just call them words instead.
}
According to \cite{Furer2014}, under these assumptions,
multiplication of two $p$-word integers can be done in time $O(p \log p)$.
The algorithm achieving this time complexity is an adaptation of the second Strassen--Schönhage algorithm 
to the word RAM model.
It works by decomposing each integer into a univariate polynomial,
use the fast Fourier transform 
to compute their convolution, then transform the result back to an integer.

Furthermore, by Kronecker substitution, 
two polynomials
with $p$-word integer coefficients and total degree $n ≤ 2^w$ can be multiplied in time $O(np \log(np))$.
This is the asymptotically best known algorithm.\footnote{
	It was mentioned in \cite[Remark~2]{stablemult} that bivariate FFT may be faster in practice,
	but only by a constant factor.
}

Let $p ≥ 1$ be an integer.
Define a $p$-word floating-point number to be a real number of the form $± m ⋅ 2^{±e}$,
where $m$ is a non-negative integer of at most $p$ words, and $e$ is an $O(1)$-word integer.
Then the problem of floating-point polynomial multiplication can be stated as follows:
\begin{quote}
	Given two polynomials with total degree $n$, where each coefficient is a $p$-word floating-point number,
	compute their product.
\end{quote}

In the easiest case, where all exponents $e$ are within $O(1)$ of each other,
the problem can be reduced to polynomial multiplication with integer coefficients
by multiplying all coefficients by an appropriate power of two.
This can be done in time $O(np)$ with bit shifts, and it makes all the coefficients $(p+O(1))$-word integers.
Then, the previously mentioned algorithm can perform the multiplication in $O(np \log(np))$.

In practice, however, it is also common for the exponents $e$ to vary.
For example, the coefficient of $x^0$ in the polynomial $(x+1)^{10000}$ is $1$,
while the coefficient of $x^{5000}$ is approximately $2^{9993}$.

In this case, it is often prohibitively expensive to compute the product exactly.
In this paper, we will describe an algorithm that solves an easier variant of the problem,
namely computing an approximation of the product with small \emph{relative Newton error}.
The time complexity of our algorithm is $O(np \log(np))$, which is the same as the easy case.

\emph{Relative Newton error} was first defined in an unpublished preprint by van der Hoeven \cite{stablemult}.
This preprint contains an algorithm that solves the same problem of approximating a floating-point polynomial product,
with the same time complexity $O(np \log(np))$.
Our algorithm has the advantage of being simpler and having smaller constant factor.

Here is how to navigate this paper.
The main algorithm is \Cref{algorithm:block_division}.
Its proof of accuracy is in \Cref{main_algorithm_accuracy},
and the proof that its time complexity is $O(np \log(np))$ is in \Cref{main_algorithm_time_complexity}.
The main idea of the algorithm, and its proof strategy, is outlined in \Cref{sec_main_algorithm}.

\section{Notation}
\label{section:notation}

\begin{table}
	\centering
	\begin{tabular}{cc}
		Symbol & Meaning \\ \hline
		$a$, $b$ & input sequences \\
		$α$, $β$ & affine function constant terms\\
		$c$ & convolution of $a$ and $b$ \\
		$d$ & difference $˘f(i)+˘g(j)-˘h(k)$ \\
		$e$ & twice-difference $˘h(k+1)$, exponent \\
		$f$, $g$ & exponent sequence of $a$ and $b$ \\
		$˘f$, $˘g$ & lower convex hull of $f$ and $g$ \\
		$˘h$ & min-plus convolution of $f$ and $g$ \\
		$i$, $j$ & indices \\
		$k$ & index, usually $i+j$ \\
		$ℓ$ & square side length \\
		$m$, $n$ & length of $a$ and $b$, mantissa \\
		$p$ & precision \\
		$q$ & index into sequence $‾k$ \\
		$𝐫$ & rectangle \\
		$σ$ & exponent slope \\
		$t$ & length of sequence $‾k$ \\
		$w$ & machine width \\
		$x$ & indeterminate, mostly unused
	\end{tabular}
	\caption{Notation.}
\end{table}

For integers $i ≤ j$, we write $[i ‥ j] = \{ i, i+1, i+2, …, j \}$.
For an integer $n ≥ 0$, we write $[n] = [0 ‥ n]$.
Sequences are represented by functions:
a sequence $a = (a(0), a(1), …, a(n))$ of real numbers is a function $a \colon [n] → ℝ$.
$a|_{[i ‥ j]}$ is the restriction of $a$ to the domain $[i ‥ j]$.

As in the introduction, a nonnegative $p$-word integer is an integer $m$ such that $0 ≤ m < 2^{w ⋅ p}$.
A $p$-word integer is an integer of the form $± m$, where $m$ is a nonnegative $p$-word integer.
We note that one more bit is needed to represent the sign.

We represent a rational $a/b$ with two (not necessarily coprime) integers $a$ and $b$.
We say $a/b$ is a $p$-word rational if the two integers $a$ and $b$ representing $a/b$
are both $p$-word integers. We note that representing $a/b$ takes at least $2p$ word, plus
one bit for the sign.

Let $𝔽$ be the set of $p$-word floating-point numbers.
This means each element of $𝔽$ has the form $± m ⋅ 2^{± e}$, where $0 ≤ m < 2^{w ⋅ p}$.
We ignore the issue of overflow and underflow, and assume all input exponents are $O(1)$-word integers,
so all output exponents are $O(1)$-word integers as well.
Note that the bit-precision of these floating-point numbers 
is $w ⋅ p$, not $p$.

Let $n$ and $m$ be nonnegative integers.
The polynomials to be multiplied have coefficients $a \colon [n] → 𝔽$ and $b \colon [m] → 𝔽$.
Their product, if computed exactly, would be $c \colon [n+m] → ℝ$
defined by $c(k) = ∑_{i+j=k} a(i)⋅b(j)$ for each $k ∈[n+m]$.
We assume $m ≤ n ≤ 2^w$.

When clear, the polynomial $∑_{i ∈[n]} a(i) x^i$ can be denoted $a$. Similar for $b$ and $c$.
The indeterminate is always $x$.
So, $a = (x+1)^2$ is shorthand for $n = 2$, $a(0) = 1$, $a(1) = 2$, $a(2) = 1$.

For any $x_0 ∈ ℂ$, the evaluation of the polynomial $∑_{i ∈[n]} a(i) x^i$
at point $x = x_0$ is written $∑_{i ∈[n]} a(i) x_0^i$. Note that $a(x_0)$
cannot be used since this notation already mean the $x_0$-th coefficient of $a$.
The subscript ``$i ∈[n]$'' can be omitted when clear.

For any sequence $a \colon [n] → ℂ$, $\log_2 |a|$ denotes the sequence $φ_a \colon [n] → ℝ ∪ \{-∞ \}$
defined by \[ φ_a (i) = \log_2 |a(i)|. \]
Similar notations are used for other operations:
for $φ \colon [n] → ℝ ∪ \{-∞ \}$, write $2^φ$ to denote the sequence $a \colon [n] → ℝ$
defined by \[ a(i) = 2^{φ(i)}. \]
For any sequences $a, b \colon [n] → ℂ$ and $ε \colon [n] → ℝ$,
write $|a-b| ≤ ε$ to mean
\[ |a(i)-b(i)| ≤ ε(i) \text{ for all } i ∈ [n]. \]

Let $\id$ denote the identity function on an unspecified domain.
With the same notation as above, for $a \colon [n] → ℂ$, $λ ∈ ℂ$,
then $a+λ ⋅ \id$ is the function that takes in $x$ and returns $a(x)+λ ⋅\id(x) = a(x)+λx$.

For any two sets $A$ and $B$, let $A × B$ be the set of all (ordered) pairs $(a, b)$
for $a ∈ A$, $b ∈ B$.
For integers $i⁻ ≤ i⁺$ and $j⁻ ≤ j⁺$,
we call a set of the form $[i⁻ ‥ i⁺] ×[j⁻ ‥ j⁺]$ a \emph{rectangle}.
If, furthermore, $i⁺-i⁻ = j⁺-j⁻$, we call that set a \emph{square}.


For real numbers $a ≤ b$, we write $[a, b] = \{ x ∈ ℝ \mid a ≤ x ≤ b \}$.
For a rectangle $R = [i⁻ ‥ i⁺] × [j⁻ ‥ j⁺]$,
we write $\Int R = [i⁻, i⁺] × [j⁻, j⁺]$.
If $R$ is a square, we also call the set $\Int R$ a square.

We will often visualize a rectangle as a set of integral points on the plane.
For example, the rectangle $[2 ‥ 5] × [1 ‥ 3]$, with $12$ elements,
is visualized in \Cref{fig_example_rectangle}.

	\begin{figure}
		\centering
		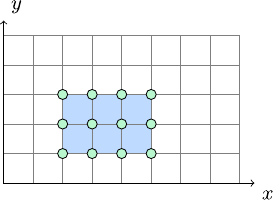
		\caption{Visualization of $[2 ‥ 5] × [1 ‥ 3]$.
	Elements of $[2 ‥ 5] × [1 ‥ 3]$ are depicted as circular points.
	The points in $\Int([2 ‥ 5] × [1 ‥ 3]) = [2, 5] × [1, 3]$ are colored blue.
}
		\label{fig_example_rectangle}
	\end{figure}

For a set $R ⊆ ℝ^2$ and $k ∈ ℝ$, define $\val(R, k)$ to be the measure of the set $\{ i ∈ ℝ \mid (i, k-i) ∈ R \}$.
We will not encounter a case where the measure does not exist or is infinite.

For a function $f$, let $\max f$ be the maximum value of $f$ on its domain.
In all the cases we work with, the maximum exists.
Similarly, $\min f$ is the minimum value of $f$ on its domain.
Write $\var f =\max f-\min f$.

\section{Prior works}
\label{sec_prior_works}

We use \cite{Brent_Zimmermann_2010} for various algorithms on arbitrary-precision floating-point numbers.
In particular, we have the following. All section references are to the arXiv preprint version of the book.
All time complexity analysis assume all input exponents are $O(1)$-word integers.

Honestly, we don't know how much detail is needed, so we err on the side of providing more detail.
A commonly-used alternative is to assume, say, multiplication of two $p$-word floating-point numbers takes time $M(p)$.
\begin{itemize}
	\item Section 1.2 shows the algorithm for integer addition and subtraction.
		While the time complexity analysis is not included there,
		it is easy to see the time complexity of adding or subtracting two $p$-word integers is $O(p)$.
		This is later used in Section 3.2 for floating-point addition and subtraction,
		which also takes $O(p)$.
	\item Section 1.3.4 discusses the application of the fast Fourier transform (FFT) in integer multiplication.
		The analysis in \cite{Furer2014}
		shows that one FFT-based algorithm can
		multiply two $p$-word integers in time $O(p \log p)$.
		As discussed in Section 3.3, this naturally lead to an $O(p \log p)$ algorithm to
		multiply two floating-point numbers.
	\item Section 3.4 gives an algorithm to divide two floating-point numbers.
		The algorithm uses Newton's iteration to compute the reciprocal of the divisor.
		Using the self-correcting property of Newton's iteration,
		the algorithm uses reduced precision at the beginning, and doubles the working precision
		after each iteration.
		Under our machine model, where word-sized integer division takes $O(1)$ time,
		an initial approximation of the reciprocal of the divisor, accurate to $Ω(w)$ bits,
		can be computed in $O(1)$ time by dividing the integers.
		Newton's iteration has quadratic convergence, so $O(\log p)$ further iterations are needed.
		Therefore, the algorithm takes time $O(p \log p)$.
	\item Section 3.5 explains how to compute the square root of a floating-point number,
		again using Newton's iteration.
		Since we do not assume a word-sized square root can be computed in $O(1)$ time,
		we need $O(\log (wp))$ iterations:
		the first $O(\log w)$ iterations take time $O(1)$ each used to compute a $Θ(w)$-bit
		approximation, and the remaining $O(\log p)$ iterations take total time
		$O(2 \log 2 + 4 \log 4 + 8 \log 8 + ⋯ + p \log p) = O(p \log p)$.

		When we already know an approximation to the result accurate to $O(w)$ bits,
		the time complexity is only $O(p \log p)$.
	\item Section 4.8 explains how to compute the logarithm of a floating-point number
		using the arithmetic--geometric mean iteration.
		Unlike Newton's iteration, the arithmetic--geometric mean iteration is not self-correcting,
		so all intermediate calculations must be done in full $p$-word precision.

		The first $O(\log w)$ iterations takes time $O(\log w+p \log p)$ each.
		The remaining $O(\log p)$ iterations takes time $O(p \log p)$ each.
		So the time complexity is $O\bigl(\log^2 w + p \log p \log (wp))\bigr)$.

		As a special case, if the required precision is only $b ≤ w$ bits,
		we need $O(\log b)$ iterations, where each iteration requires calculating a square root to $b$ bits,
		which takes $O(\log b)$. So the time complexity is $O(\log^2 b)$.
	\item Using the fast algorithm for computing the logarithm as above
		together with Newton's method, described in Section 4.2.5,
		we can compute the exponential of a floating-point number
		in time $O\bigl((\log^2 w + p \log p) \log(wp)\bigr)$.
		We believe the $O(\log^3 w)$ term can be made smaller by,
		for example, store intermediate values used in an arithmetic--geometric mean computation
		as the initial approximation in the next arithmetic--geometric mean computation,
		but in practice this is not the dominant term in the time complexity.
\end{itemize}

Intermediate values in a calculation often need to be rounded to a floating-point number.
Adapting \cite[Theorem~3.1.1]{Brent_Zimmermann_2010} to the notation in this article, we have:
\begin{lemma}
	For any $x ∈ ℝ$, there exists $u ∈ 𝔽$ such that $\frac{|u-x|}{|x|} ≤ 2^{-w⋅p}$.
\end{lemma}
Because of the table maker's dilemma, we will be more relaxed in our error analysis
and assume that all operations have relative error $≤ 2^{1-w⋅p}$.
In other words, if $x$ is the exact result, and $u$ is the approximate result,
then
\begin{equation}
	\label{error_requirement}
	\frac{|u-x|}{|x|} ≤ 2^{1-w⋅p}.
\end{equation}
Real implementations often have exactly-rounded results for elementary operations,
and $\frac{|u-x|}{|x|} ≤ (0.5+ε) ⋅ 2^{-w⋅p}$ for transcendental functions, for a small constant $ε$. 

\section{Related problems}
\label{sec_related_problems}

As mentioned in the introduction, our algorithm can be applied for the near-convex min-plus convolution problem too.
\begin{definition}[Min-plus convolution]
	Let $n ≥ m$ be nonnegative integers.
	Let $f \colon [n] → ℤ$ and $g \colon [m] → ℤ$ be two sequences of $O(1)$-word integers.
	The \emph{min-plus convolution} of $f$ and $g$ is a sequence $h \colon [n+m] → ℤ$ defined by
	\[ h(k) = \min_{\substack{(i,j) ∈ [n] ×[m] \\i+j=k}} f(i)+g(j). \]
\end{definition}
In general, when $n ∈ Θ(m)$, under the min-plus convolution hypothesis \cite{kunnemann_et_al:LIPIcs.ICALP.2017.21,cygan_et_al:LIPIcs.ICALP.2017.22}, we cannot do better than $n^{2-o(1)}$.

However, if we assume the sequences $f$ and $g$ are \emph{$Δ$-near convex}, an algorithm with time complexity $Õ(nΔ)$ was presented in \cite{bringmann_2023}.
Following the analysis, we see that the algorithm there takes time $O(nΔ \log n \log (nΔ))$.
Adapting our algorithm to this problem, we obtain an algorithm with time complexity $O(nΔ \log (nΔ))$,
which is faster by a logarithmic factor.

\section{Linear scaling}

An algorithm to multiply polynomials 
whose coefficient magnitudes approximately follow a geometric progression
was described in \cite[Section 3]{stablemult}.
We will use the following result, which uses the same underlying idea.

\begin{proposition}
	\label{prop:linear_scaling}
	Let $0 ≤ i⁻ ≤ i⁺ ≤ n$ and $0 ≤ j⁻ ≤ j⁺ ≤ m$ be integers.
	Suppose $m ≤ n$.
	Let $𝐫 = [i⁻‥ i⁺] × [j⁻‥ j⁺]$ be a rectangle.
	Let $a \colon [i⁻ ‥ i⁺] → 𝔽$ and $b \colon [j⁻ ‥ j⁺] → 𝔽$.
	Let $c \colon [i⁻+j⁻ ‥ i⁺+j⁺] → ℝ$ be their convolution computed exactly:
	\[ c(k) = ∑_{\substack{(i, j) ∈ 𝐫 \\ i+j=k}} a(i) ⋅ b(j).  \]
	Let $σ ∈ ℚ$ be a $O(1)$-word rational,
	and $α, β ∈ ℤ$ be $O(1)$-word integers.
	Suppose $\RescaleRound{⋅}$ is implemented such that
	\begin{equation}
		\label{error_requirement_rescale_round}
		|‾a(i) - a(i)/2^{σ ⋅ i+α}| ≤ \min\bigl(1, |a(i)|/2^{σ ⋅ i+α+1}\bigr),
	\end{equation}
	and $\Unscale{⋅}$ is implemented such that
	\begin{equation}
		\label{error_requirement_unscale}
		|˜c(k) - ‾c(k)⋅2^{σ ⋅ k+α}| ≤ 2^{σ ⋅ k+α + 1-w⋅p} ⋅ |‾c(k)| .
	\end{equation}
	Further suppose there are functions $φ_a \colon [i⁻ ‥ i⁺] → ℝ$ and $φ_b \colon [j⁻ ‥ j⁺] → ℝ$
	such that
	\begin{equation}
		\label{a_le_2_pow_phi}
		\begin{aligned}
			\max\bigl(|a(i)|, 2^{σi+α+wp}\bigr) ≤ 2^{-φ_a (i)}, \\
			\max\bigl(|b(j)|, 2^{σj+β+wp}\bigr) ≤ 2^{-φ_b (j)}.
		\end{aligned}
	\end{equation}
	Let $φ_c \colon [i⁻+j⁻ ‥ i⁺+j⁺] → ℝ$ be the min-plus convolution of $φ_a$ and $φ_b$:
	\[ φ_c (k) = ∑_{\substack{(i,j) ∈ 𝐫 \\ i+j=k}} φ_a (i)+φ_b(j). \]
	Then, for each $k ∈ [i⁻+j⁻ ‥ i⁺+j⁺]$, $˜c(k)$ computed by \Cref{algorithm:linear_scaling}
	satisfies
	\[ 
		|˜c(k) - c(k)| ≤ 22t ⋅ 2^{-wp}⋅ 2^{-φ_c (k)} ,
	\]
	where $t = \val(\Int 𝐫, k)+1$
	is the number of terms in the sum
	$∑_{\substack{(i, j) ∈ 𝐫 \\ i+j=k}} a(i) b(j)$.
\end{proposition}

\begin{algorithm}
	\caption{Polynomial multiplication with linear scaling.}
	\label{algorithm:linear_scaling}
	\begin{algorithmic}[1]
		\Procedure{RescaleRound}{$a \colon [i⁻ ‥ i⁺] → ℝ, σ, α$}
			\For{$i ∈ [i⁻ ‥ i⁺]$}
				\State $‾a(i) ← \round (a / 2^{σ ⋅ i+α})$
			\EndFor
			\State \Return $‾a$
		\EndProcedure
		\Procedure{Unscale}{$‾c \colon [k⁻ ‥ k⁺] → ℝ, σ, α$}
			\For{$k ∈ [k⁻ ‥ k⁺]$}
				\State $˜c(k) ← ‾c(k) ⋅ 2^{σ ⋅ k+α}$
			\EndFor
			\State \Return $˜c$
		\EndProcedure
		\Procedure{RescaledConvolution}{$a, b, 𝐫 = [i⁻ ‥ i⁺] × [j⁻ ‥ j⁺], σ, α, β$}
			\State $‾a ← \RescaleRound{a|_{[i⁻ ‥ i⁺]}, σ, α}$
			\State $‾b ← \RescaleRound{b|_{[j⁻ ‥ j⁺]}, σ, β}$
			\State $‾c ← \IntegerConvolution{‾a, ‾b}$
			\State $˜c ← \Unscale{‾c, σ, α+β}$
		\EndProcedure
	\end{algorithmic}
\end{algorithm}


While both the proposition statement and proof are extremely technical,
the following is the easiest example.
Let $a \colon [n] → ℤ$ and $b \colon[m] → ℤ$
such that for each $i ∈[n]$, $j ∈[m]$, $2^{wp-1}≤ a(i) < 2^{wp}$
and $2^{wp-1}≤ b(j) < 2^{wp}$.

To apply \Cref{prop:linear_scaling}, we need to pick $φ_a$ and $φ_b$
satisfying \Cref{a_le_2_pow_phi}.
For this, $φ_a (i) = φ_b (j) = -p$ for all $i$ and $j$, and $σ=α=β=0$ would suffice.
Then, we get $φ_c (k) = -2p$, and the absolute error
$|˜c(k) - c(k)|$ has the order of magnitude around $22 (m+1) ⋅ 2^{2wp - wp} ∼ m ⋅ 2^{wp}$.
We expect $c(k) ∼ m ⋅ 2^{2wp}$, so this is a relative error of $∼ 2^{-wp}$.
In other words, we get $p$ words of precision, which is just what we expect.

\begin{remark}
	\label{remark_loose_round_requirement}
We note that if $\round(⋅)$ is exactly implemented,
we would have $|‾a(i) - a/2^{σ ⋅ i+α}| ≤ 0.5$.
We pick a looser bound for the same reason as discussed after \Cref{error_requirement}.


\Cref{error_requirement_unscale} is taken directly from \Cref{error_requirement},
but to simplify the analysis, we consider the whole computation of
$‾c(k)⋅2^{σk+α}$ as one operation for the purpose of error analysis.

The constant factor, $22$, is likely extremely pessimistic.
However, in order to show the asymptotic order of the error term,
this will suffice.
\end{remark}

\begin{proof}
	We first consider $‾c$. Since $\IntegerConvolution{⋅}$ is exact,
	we have, for each $k$,
	\begin{equation}
		\label{equation_bar_c}
		‾c(k) 2^{σk+α+β} = ∑_{\substack{(i, j) ∈ 𝐫 \\ i+j=k}} ‾a(i) ‾b(j) 2^{σk+α+β}.
	\end{equation}

	From \Cref{error_requirement_rescale_round},
	we have $|‾a(i)| ≤ 3 ⋅ |a(i)|/2^{σi+α}$.
	so from \Cref{a_le_2_pow_phi},
	$|‾a(i)| ≤ 3 ⋅ 2^{-φ_a (i) - (σi+α)}$.

	From \Cref{a_le_2_pow_phi}, 
	$|b(j)| ≤ 2^{-φ_b (j)}$.

	From \Cref{error_requirement_rescale_round}, we also have
	\begin{align*}
		| ‾a(i) 2^{σi+α} - a(i) | &≤ 2^{σi+α}, \\
		| ‾b(j) 2^{σk+α+β} - b(j) 2^{σi+α} | &≤ 2^{σk+α+β}.
	\end{align*}

	Therefore,
	\begin{align*}
		& \quad | ‾a(i) ‾b(j) 2^{σk+α+β} - a(i) b(j) | \\
		&≤ | ‾a(i) ‾b(j) 2^{σk+α+β} - ‾a(i) b(j) 2^{σi+α} | + | ‾a(i) b(j) 2^{σi+α} - a(i) b(j) | \\
		&= |‾a(i)| ⋅ | ‾b(j) 2^{σk+α+β} - b(j) 2^{σi+α} | + |b(j)|⋅ | ‾a(i) 2^{σi+α} - a(i) | \\
		&≤ 3 ⋅ 2^{-φ_a (i) - (σi+α)} ⋅ 2^{σk+α+β} + 2^{-φ_b (j)} ⋅ 2^{σi+α} \\
		&= 3 ⋅ 2^{-φ_a (i) + σj+β} + 2^{-φ_b (j) + σi+α} \\
		&≤ 3 ⋅ 2^{-φ_a (i) -wp-φ_b (j)} + 2^{-φ_b (j) -wp+φ_a (i)} \\
		&= 4 ⋅ 2^{-φ_a (i)-φ_b(j) -wp}.
	\end{align*}

	Since the sum on the right hand side of \Cref{equation_bar_c} has $t$ terms,
	we have
	\begin{equation}
		\label{error_bound_part1}
		|‾c(k) 2^{σk+α+β} - c(k)| ≤ 2^{-wp} ⋅ 4t ⋅ 2^{-φ_c (k)}.
	\end{equation}

	It remains to look at the rounding error when computing $˜c(k)$ from $‾c(k)$.
	By the assumption in \Cref{error_requirement_unscale}, we have
	\[ |˜c(k)-2^{σk+α+β} ‾c(k) | ≤ 2^{σk+α+β} ‾c(k) ⋅ 2^{1-wp}. \]

	We already derived above that $|‾a(i)| ≤ 3 ⋅ 2^{-φ_a (i)-(σi+α)}$.
	Similarly, $|‾b(j)| ≤ 3 ⋅ 2^{-φ_b (j)-(σj+β)}$.
	Therefore, from \Cref{equation_bar_c},
	\[
		‾c(k) 2^{σk+α+β} ≤ 9t ⋅ 2^{-φ_c (k)},
	\]
	and
	\[ |˜c(k)-2^{σk+α+β} ‾c(k) | ≤ 9t ⋅ 2^{1-wp-φ_c (k)}. \]
	Combining with \Cref{error_bound_part1}, we get
	\[ |˜c(k)-c(k) | ≤ 22t ⋅ 2^{-wp-φ_c (k)}. \tag*{\qedhere} \]
\end{proof}

We also want to analyze the time complexity of \Cref{algorithm:linear_scaling}.
\begin{proposition}
	With notations and assumptions as in \Cref{prop:linear_scaling}.
	Suppose for each $(i, j) ∈ 𝐫$, $-φ_a (i) - (σi+α) ∈ O(p)$ and $-φ_b (j) - (σj+β) ∈ O(p)$.
	Then \Cref{algorithm:linear_scaling} can be implemented to run in
	time $O\bigl(n p \log (np) + (\log w+p \log p) \log(wp)\bigr)$.
\end{proposition}
We note that, when $n$ is reasonably large,
the $np \log (np)$ term dominates the $(\log w+p \log p) \log(wp)$ term,
so the time complexity is closer to $O(np \log (np))$ in practice.

\begin{proof}
	Recall from \Cref{sec_prior_works} that the logarithm and the exponential can be computed to $p$ words of precision
	in time $O\bigl((\log^2 w+p \log p)\log(wp)\bigr)$.
	So, for a $O(1)$-word rational $σ$,
	$2^σ$ can be computed to $p+1$ words of precision in time $O\bigl((\log^2 w+p \log p)\log(wp)\bigr)$.
\end{proof}

\section{The main algorithm}

\subsection{Overview}

Let $a \colon [n] → 𝔽$ and $b \colon [m] → 𝔽$ be as in \Cref{section:notation}.
If either $a$ or $b$ are all zeroes, the product is zero.
Otherwise, if any of $a(0)$, $a(n)$, $b(0)$ or $b(m)$ is zero,
the algorithm can delete as many zeroes as possible from the two ends of the two coefficient lists,
compute the convolution of the remaining part, and put back the zeroes.
Therefore, without loss of generality, we can assume none of $a(0)$, $a(n)$, $b(0)$ or $b(m)$ is zero.

\begin{definition}[Definition of $f$, $g$, $˘f$, $˘g$, $˘h$]
Let $f \colon [n] → ℤ ∪ \{+∞ \}$ be defined by $f = -⌈\log_2 |a|⌉$.
Similarly, $g \colon [m] → ℤ ∪ \{+∞ \}$, $g = -⌈\log_2 |b|⌉$.
Let $˘f \colon [n] → ℚ$ be the lower convex hull of $f$.
Note that since neither $a(0)$ nor $a(n)$ is zero, none of the elements in $˘f$ is $-∞$.
Define $˘g$ similarly.
Let $˘h \colon [n+m] → ℚ$ be the min-plus convolution of $˘f$ and $˘g$.
\end{definition}

\begin{definition}[Definition of $d(i, j)$]
For each $i ∈[n]$ and $j ∈[m]$, define $d(i,j)= ˘f(i) + ˘g(j) - ˘h(i+j)$.
\end{definition}
By the definition of $˘h$, $d(i,j)≥ 0$.

\begin{definition}[Definition of $i⁻_k$, $i⁺_k$, $j⁻_k$, $j⁺_k$]
	\label{def_ij_plus_minus_k}
For each $k ∈ [n+m]$, consider the set of pairs $(i, j)$ such that
$i+j=k$ and $d(i,j) ≤ w⋅p$.
Note that at least one such pair exist, since by the definition of $h$,
there exists $(i,j)$ such that $i+j=k$ and $d(i,j) = 0$.
Let $(i⁻_k, j⁺_k)$ be the pair with the smallest $i$ (equivalently, largest $j$),
and $(i⁺_k, j⁻_k)$ be the pair with the largest $i$ (equivalently, smallest $j$).
\end{definition}

The high-level idea is the following.
For each $k$, for some $i₁ ≤ i⁻_k ≤ i⁺_k ≤ i₂$,
we will approximately compute $˜c(k) = ∑_{i ∈[i₁ ‥ i₂]} a(i) b(k-i)$.
This is meant to be an approximation to $c(k) = ∑_{i+j=k} a(i) b(j)$, with an error of:
\begin{align}
	|c(k) - ˜c(k)| &≤ ∑_{i ∉ [i₁ ‥ i₂]} |a(i) b(k-i)| \nonumber \\
				   &≤ (m+1) ⋅ \max_{i ∉ [i₁ ‥ i₂]} 2^{-f(i)-g(k-i)} \nonumber \\
				   &≤ (m+1) ⋅ \max_{i ∉ [i₁ ‥ i₂]} 2^{-˘f(i)-˘g(k-i)} \nonumber \\
				   &< (m+1) ⋅ 2^{-˘h(k) - w ⋅ p}.  \label{tilde_c_error_bound}
\end{align}
If we optimistically expect that
\begin{equation}
	\label{optimistic_c_bound}
\begin{aligned}
|c(k)| &≥ Ω(1) ⋅ \max_{i+j=k} |a(i) ⋅ b(j)| \\
\max_{i+j=k} |a(i) ⋅ b(j)| &≥ Ω(1) ⋅ 2^{-˘h(k)},
\end{aligned}
\end{equation}
then we get that the relative error of $˜c(k)$ is bounded by
\[ \frac{|c(k) - ˜c(k)|}{|c(k)|} ≤ (m+1) ⋅ 2^{-w ⋅ p} ⋅ O(1), \]
which is quite good.

Later, in \Cref{sec_discussion},
we will explain why,
even when \Cref{optimistic_c_bound} is not satisfied,
the result, with error estimated in \Cref{tilde_c_error_bound}, remains useful.

Without making any assumption, we make the following definition to formalize the error bound of the result.
\begin{definition}[{Relative Newton error}]
	\label{def_relative_Newton_error}
	Let $a$, $b$, $f$, $g$, $˘f$, $˘g$ and $˘h$ be as above.
	Let $c \colon [n+m] → ℝ$ be the (exactly computed) convolution of $a$ and $b$.
	For any $˜c \colon [n+m] → ℝ$, the relative Newton error of $˜c$ with respect to $c$ is defined to be
	\[ \max_{k ∈[n+m]} \frac{|˜c(k)-c(k)|}{2^{-˘h(k)}}. \]
\end{definition}
Usually, $˜c$ would be the output of our algorithm, to be described below.
The definition can be seen to be similar to the standard definition of a relative floating-point error:
if $c ∈ ℝ ∖ \{ 0 \}$ is some exact value and $˜c ∈ 𝔽$ is an approximation of $c$,
then the relative error of $˜c$ is simply $\frac{|˜c-c|}{|c|}$.
If $2^{-˘h(k)} ≈ |c(k)|$, then the relative Newton error of $˜c$
with respect to $c$ is approximately the maximum relative error of any coefficient $˜c(k)$ with respect to $c(k)$.

We note that our definition of relative Newton error is slightly different from \cite[Section~5]{stablemult}.
This will be discussed in \Cref{sec_discuss_def_relative_Newton_error}.

\subsection{Applying linear scaling to a subrectangle}

The following proposition states that, if $𝐫 = [i⁻ ‥ i⁺] × [j⁻ ‥ j⁺]$ is a rectangle whose longer side is no more than $λ$ times longer than the shorter side, then both functions $˘f|_{[i⁻ ‥ i⁺]}$ and $˘g|_{[j⁻ ‥ j⁺]}$ are within $L¹$ distance of $λδ$ from certain linear functions.

\begin{proposition}
	\label{prop_linear_scaling_on_rectangle}
	Let $𝐫 = [i⁻ ‥ i⁺] × [j⁻ ‥ j⁺]$ be a rectangle.
	Suppose $1 ≤ i⁺-i⁻ ≤ j⁺-j⁻ ≤ λ (i⁺-i⁻)$
	or $1 ≤ j⁺-j⁻ ≤ i⁺-i⁻ ≤ λ (j⁺-j⁻)$
	for a positive integer constant $λ$.
	Furthermore suppose $\max_{(i, j) ∈ 𝐫} d(i, j) ≤ δ$ for a positive integer $δ$.

	Define:
	\begin{gather*}
		σ = \frac{˘f(i⁻)-˘f(i⁺)}{i⁺-i⁻}, \\
		α = -˘f(i⁻)-σi⁻-wp, \\
		β = \min\bigl(-˘g(j⁺)-σj⁺, -˘g(j⁻)-σj⁻\bigr)-wp.
	\end{gather*}

	Then $α$, $β$ and $σ$ are $O(1)$-word rationals 
	such that $σi+α+wp ≤ -˘f(i)$ and $σj+β+wp ≤ -˘g(j)$
	for all $(i, j) ∈ 𝐫$.
	Furthermore,
	\begin{equation}
		\label{discrete_coefficient_magnitude_bound}
		-˘f(i) ≤ σi+α+wp+λδ \text{ and } -˘g(j) ≤ σj+β+wp+λδ
	\end{equation}
	for all $(i, j) ∈ 𝐫$.
\end{proposition}
For such a rectangle $𝐫$, we intend to apply \Cref{prop:linear_scaling} afterwards.
The conclusion of this proposition corresponds to the hypothesis \Cref{a_le_2_pow_phi}.

We need a few technical lemmas.

\begin{lemma}
	\label{bound_var_sum_1}
	Let $i⁻ ≤ i⁺$ be integers.
	Let $f, g \colon [i⁻ ‥ i⁺] → ℝ$ be convex functions.
	Suppose $f(i⁺) ≤ f(i⁻)$ and $g(i⁺) ≤ g(i⁻)$.
	Then
	\begin{equation}
		\label{var_bound_1}
		\max(\var f,\var g) ≤ \var(f+g).
	\end{equation}
\end{lemma}
\begin{proof}
	Without loss of generality, by adding a constant to $f$ and $g$,
	we can assume $f(i⁻)=g(i⁻)=0$.
	Then $f(i⁺) ≤ 0$ and $g(i⁺) ≤ 0$.
	By the convexity of $f$ and $g$, $f(i) ≤ 0$ and $g(i) ≤ 0$ for every $i ∈ [i⁻ ‥ i⁺]$.
	Therefore, $\max(-\min f,-\min g)≤-\min(f+g)$.
	Since $\max f =\max g =\max(f+g)=0$, this is equivalent to \Cref{var_bound_1}.
\end{proof}

\begin{lemma}
	\label{bound_var_sum_2}
	Let $i⁻ ≤ i⁺$ be integers.
	Let $f, g \colon [i⁻ ‥ i⁺] → ℝ$ be convex functions.
	Suppose $f(i⁺) ≥ f(i⁻)$ and $g(i⁺) ≥ g(i⁻)$.
	Then
	\[ \max(\var f,\var g) ≤ \var(f+g).  \]
\end{lemma}
\begin{proof}
	Identical to the proof of the previous lemma, interchanging the role of $i⁻$ and $i⁺$.
\end{proof}

\begin{lemma}
	\label{lemma_common_slope}
	Let $i⁻ ≤ i⁺$ be integers.
	Let $f, g \colon [i⁻ ‥ i⁺] → ℝ$ be convex functions.
	Let $σ$ be any value between $\frac{f(i⁺)-f(i⁻)}{i⁺-i⁻}$ and $\frac{g(i⁻)-g(i⁺)}{i⁺-i⁻}$.
	Then \[ \max\bigl(\var (f-σ ⋅ \id), \var (g+σ ⋅ \id)\bigr) ≤ \var (f+g). \]
\end{lemma}
See \Cref{fig_common_slope} for an illustration.
We require $i⁻<i⁺$ only for $σ$ to be well-defined. When $i⁻=i⁺$, the statement is true for every $σ$.

	\begin{figure}
		\centering
		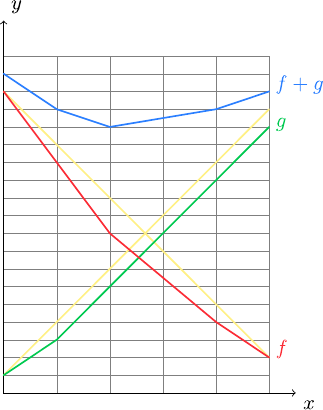
		\caption{
	Illustration for the proof of \Cref{lemma_common_slope}.
	Here $f \colon [5] → ℝ$ is defined by $(f(0),…, f(5)) = (17, 13, 9, 6.5, 4, 2)$
	and $g \colon [5] → ℝ$ is defined by $(g(0),…, g(5)) = (1, 3, 6, 9, 12, 15)$.
	We have $\var (f+g) = 3$.
	The lemma's conclusion implies both $f$ and $g$ are well approximated by affine functions,
	drawn as yellow straight lines.
}
		\label{fig_common_slope}
	\end{figure}

\begin{proof}
	Let $‾f = f-σ ⋅\id$
	and $‾g = g+σ ⋅\id$.
	Note that $f+g=‾f+‾g$.
	Since $f$ and $g$ are both convex, so are $‾f$ and $‾g$.

	If $\frac{f(i⁺)-f(i⁻)}{i⁺-i⁻} ≤ σ ≤ \frac{g(i⁻)-g(i⁺)}{i⁺-i⁻}$,
	then $f(i⁺)-σ ⋅ i⁺ ≤ f(i⁻)-σ ⋅ i⁻$
	and $g(i⁺)+σ ⋅ i⁺ ≤ g(i⁻)+σ ⋅ i⁻$,
	we apply \Cref{bound_var_sum_1}.

	Otherwise, $\frac{f(i⁺)-f(i⁻)}{i⁺-i⁻} ≥ σ ≥ \frac{g(i⁻)-g(i⁺)}{i⁺-i⁻}$,
	we apply \Cref{bound_var_sum_2}.
\end{proof}

\begin{proof}[Proof of \Cref{prop_linear_scaling_on_rectangle}]
	Without loss of generality, we handle the case $1 ≤ i⁺-i⁻ ≤ j⁺-j⁻ ≤ λ (i⁺-i⁻)$.
	The opposite case is similar.

	From the definition of $α$, it easily follows that $α = -˘f(i⁺)-σi⁺-wp$ as well.
	By the convexity of $˘f$ and $˘g$, we
	immediately get
	\begin{equation}
		\label{lower_bound_minus_fg}
		σi+α+wp ≤ -˘f(i)\text{ and }σj+β+wp ≤ -˘g(j)
	\end{equation}
	for all $(i, j) ∈ 𝐫$.

	It remains to prove \Cref{discrete_coefficient_magnitude_bound}.
	Since \Cref{lower_bound_minus_fg} is attained for at least one $(i, j) ∈ 𝐫$,
	it suffices for us to prove 
	\[ \var (-˘f-σ ⋅\id) ≤ λδ\text{ and } \var (-˘g-σ ⋅\id) ≤ λδ. \]

	First, we prove the case where $λ = 1$, that is, $𝐫$ is a square.
	Set $k = i⁺+j⁺$. Since $𝐫$ is a square, $k = i⁻+j⁻$ as well.

	Consider the function $η \colon [j⁻ ‥ j⁺] → ℝ$,
	$η(j) = ˘f(k-j) + ˘g(j)$.
	By the definition of $˘h$, $η(j) = ˘h(k) + d(k-j,j)$.
	From the assumptions, we have $0 ≤ d(k-j,j) ≤ δ$ for every $j ∈[j⁻ ‥ j⁺]$.
	Therefore, $\var η ≤ δ$.

	By \Cref{lemma_common_slope},
	we have $\var (˘f+σ ⋅\id) ≤ δ$ and $\var (˘g+σ ⋅\id) ≤ δ$.
	The result immediately follows.

	Now suppose $λ>1$.
	Without loss of generality, suppose $λ = ⌈\frac{j⁺-j⁻}{i⁺-i⁻}⌉$.
	We note that the definition of $σ$ and $α$ is independent of $j⁻$ and $j⁺$.
	
	Set $ℓ = i⁺-i⁻ ≥ 1$.
	We split the range $[j⁻ ‥ j⁺]$ into $λ$ overlapping ranges,
	each has $ℓ+1$ elements. Specifically,
	for $p ∈ [1 ‥ λ-1]$, set $[j⁻_p ‥ j⁺_p] = [j⁻+pℓ-ℓ ‥ j⁻+pℓ]$,
	and for $p = λ$, set $[j⁻_λ ‥ j⁺_λ] = [j⁺-ℓ ‥ j⁺]$.
	Then, for each $p ∈ [1 ‥ λ]$,
	the rectangle $[i⁻ ‥ i⁺] × [j⁻_p ‥ j⁺_p]$ is a square,
	so the proof above shows $\var(˘g+σ ⋅ \id)|_{[j⁻_p ‥ j⁺_p]} ≤ δ$.
	Since the ranges $[j⁻_p ‥ j⁺_p]$ are overlapping,
	we therefore have $\var(˘g+σ ⋅ \id) ≤ λδ$,
	and the result follows.
\end{proof}

\subsection{The algorithm}
\label{sec_main_algorithm}

Here is how the values $˜c(k)$ defined above can be approximated.
Let $𝐫_k$ be the square $[i⁻_k ‥ i⁺_k] × [j⁻_k ‥ j⁺_k]$.
Let 
\begin{equation} \label{kappa_ge_4} κ>4 \end{equation}
be a constant.
Let $t$ be a positive integer that we will select later.
We iteratively define a sequence $‾k \colon [t] → ℤ$ as follows.
First, $‾k(0) = 0$.
Then, for each $q ∈ \{ 0, 1,… \}$, when $‾k(q)$ is already defined and $‾k(q) ≤ n+m$,
we define $‾k(q+1)$ as the largest integer such that
$‾k(q) < ‾k(q+1) ≤ n+m+1$ and
\begin{equation}
	\label{max_d_corner_le}
	\max\Bigl(d\bigl(i⁻_{‾k(q)},j⁺_{‾k(q+1)-1}\bigr),d\bigl(i⁺_{‾k(q+1)-1},j⁻_{‾k(q)}\bigr)\Bigr) ≤ κ ⋅ w ⋅ p.
\end{equation}
This is well-defined because, when $‾k(q+1) = ‾k(q)+1$,
by the definitions of $(i⁻,j⁺)$ and $(i⁺,j⁻)$,
\[ \max\Bigl(d\bigl(i⁻_{‾k(q)},j⁺_{‾k(q+1)-1}\bigr),d\bigl(i⁺_{‾k(q+1)-1},j⁻_{‾k(q)}\bigr)\Bigr) ≤ w ⋅ p, \]
which is $<κ ⋅ w ⋅ p$ because $κ>1$.

The sequence stops at the first entry equal to $n+m+1$.
In other words, $t$ is selected such that $‾k(t) = n+m+1$.
Because the sequence is strictly increasing, it terminates after finitely many terms.
More specifically, $1 ≤ t ≤ n+m+1$.

Given the sequence $‾k$, we use \Cref{algorithm:block_division}
to compute the values $˜c(k)$ for all $k ∈[n+m]$.

\begin{algorithm}
	\caption{The definition of \textproc{AutoRescaledConvolution}.}
	\label{algorithm:auto_rescaled_convolution}
	\begin{algorithmic}[1]
		\Procedure{AutoRescaledConvolution}{$a, b, k⁻, k⁺$}
			\State $i⁻ ← i⁻_{k⁻}$
			\State $i⁺ ← i⁺_{k⁺}$
			\State $j⁻ ← j⁻_{k⁻}$
			\State $j⁺ ← j⁺_{k⁺}$
			\State $σ ← \frac{˘f(i⁻)-˘f(i⁺)}{i⁺-i⁻}$
			\State $α ← -˘f(i⁻)-σi⁻-wp$
			\State $β ← \min\bigl(-˘g(j⁺)-σj⁺, -˘g(j⁻)-σj⁻\bigr)-wp$  \Comment{As in \Cref{prop_linear_scaling_on_rectangle}}
			\State $‾c ← \RescaledConvolution{a, b, [i⁻‥ i⁺] ×[j⁻‥ j⁺], σ, α, β}$
			\State \Return $‾c|_{[k⁻ ‥ k⁺]}$
		\EndProcedure
	\end{algorithmic}
\end{algorithm}

	\begin{figure}
		\centering
		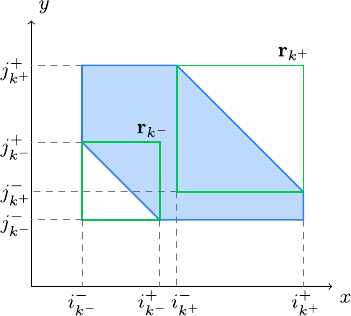
		\caption{
	Illustration for a single invocation of $\textproc{AutoRescaledConvolution}$.
}
		\label{fig_auto_rescaled_convolution}
	\end{figure}

\begin{algorithm}
	\caption{The main algorithm. Computation of the values $˜c(k)$.}
	\label{algorithm:block_division}
	\begin{algorithmic}[1]
		\Procedure{ApproximateConvexConvolution}{$a, b, ‾k$}
			\For{$q ∈ [t-1]$}
				\State $k⁻ ← ‾k(q)$
				\State $k⁺ ← ‾k(q+1)-1$
				\State $‾c ← \AutoRescaledConvolution{a, b, k⁻, k⁺}$
				\For{$k ∈ [k⁻ ‥ k⁺]$}
					\State $˜c(k) ← ‾c(k)$
				\EndFor
			\EndFor
			\State \Return $˜c$
		\EndProcedure
	\end{algorithmic}
\end{algorithm}

	\begin{figure}
		\centering
		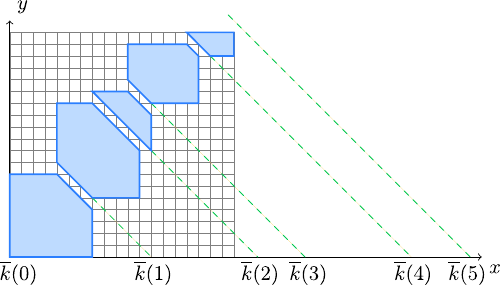
		\caption{
	Illustration for \Cref{algorithm:block_division}, the main algorithm.
}
		\label{fig_main_algorithm}
	\end{figure}

Visually, the construction of the sequence $‾k$,
and the subsequent invocations of $\AutoRescaledConvolution{a, b, k⁻, k⁺}$
can be understood as finding a collection of rectangles
that together cover all points $(i, j)$ with $d(i, j) ≤ w ⋅ p$.
See \Cref{fig_main_algorithm} for an illustration.

To explain what this illustration means:
Consider a single blue polygon.
See \Cref{fig_auto_rescaled_convolution}.
This polygon represents a single invocation of the procedure $\AutoRescaledConvolution{a, b, k⁻, k⁺}$,
with the values $i⁻_{k⁻}$, $i⁺_{k⁻}$, $i⁻_{k⁺}$, etc. marked on the diagram.
Since we only need the output values at indices $k$ satisfying $k⁻ ≤ k ≤ k⁺$,
we don't draw the whole rectangle, but rather, we cut away two triangles:
the first being $\{(i, j) ∈ 𝐫_{k⁻}, i+j<k⁻\}$,
and the second being $\{(i, j) ∈ 𝐫_{k⁺}, i+j>k⁺\}$.
For each integral point $(i, j)$ in the drawn polygon, the product $a(i) b(j)$ contributes to
the calculation of $˜c(i+j)$.
The time complexity of such a single invocation is roughly linearithmic in the \emph{perimeter} of the polygon,
as we will prove in \Cref{assume_rescaled_convolution}.
So, the total time complexity of the algorithm is roughly linearithmic in the total perimeter of all polygons.
This total perimeter is proved to be linear in $n+m$ in \Cref{sum_side_rectangles_le}.

The construction of the sequence $‾k$ as described above
greedily constructs the blue polygons from the bottom-left corner to the top-right corner.
From \Cref{max_d_on_corner},
it is clear that the resulting collection covers all points $(i, j)$ satisfying $d(i, j) ≤ wp$
and does not cover any $(i, j)$ satisfying $d(i, j) > κwp$.
This fact is a prerequisite to apply \Cref{assume_rescaled_convolution}.

\subsection{Technical lemmas}

We will need to prove the correctness and time complexity
of this algorithm.
To do so, we need several technical lemmas.

For $k ∈ [1 ‥ n+m-1]$, define $e(k) = ˘h(k+1) - 2 ˘h(k) + ˘h(k-1)$.
Because $˘h$ is convex, $e(k) ≥ 0$ for every $k$.

\begin{theorem}
	\label{main_algorithm_accuracy}
	The relative Newton error of the sequence $˜c$ returned by the algorithm
	$\textproc{ApproximateConvexConvolution}$ is at most $23(m+1) 2^{-wp}$.
\end{theorem}

As explained in \Cref{prop:linear_scaling}, the constant factor is likely to be much better in practice.

In a high level, there are only two sources of error: of omitting terms in the sum, and of floating-point error.
The former is bounded by formalizing \Cref{tilde_c_error_bound},
and the latter is handled by applying \Cref{prop:linear_scaling}.

\begin{proof}
	Fix any $k ∈ [n+m]$. 
	Expanding \Cref{def_relative_Newton_error},
	we need to prove
	\[ |˜c(k)-c(k)| ≤ 23(m+1) ⋅ 2^{-wp-˘h(k)}. \]

	Let $q$ be the unique value of $q ∈ [t-1]$ satisfying $‾k(q) ≤ k < ‾k(q+1)$.
	Let $k⁻ = ‾k(q)$ and $k⁺ = ‾k(q+1)-1$.
	Let $𝐫 = [i⁻_{k⁻}‥ i⁺_{k⁺}]× [j⁻_{k⁻}‥ j⁺_{k⁺}]$.
	By \Cref{prop:linear_scaling}, we have
	\[ \abs[\bigg]{ ˜c(k)-∑_{\substack{(i,j) ∈ 𝐫 \\ i+j=k}} a(i) b(j) } ≤ 22 (m+1) 2^{-wp-˘h(k)}. \]

	Recall, by the definition of $‾k$, $k⁻ ≤ k ≤ k⁺$.
	From \cite[Lemma~18, Monotone Paths]{bringmann_2023_arxiv}
	and the definitions of the sequences $i⁻_k$, $i⁺_k$, $j⁻_k$, $j⁺_k$
	in \Cref{def_ij_plus_minus_k},
	it follows that, for all $(i, j) ∈ [n]×[m] ∖ 𝐫$, $i+j=k$, then $d(i, j) > w⋅p$.
	This means
	\begin{align*}
		\abs[\bigg]{ c(k) - ∑_{\substack{(i,j) ∈ 𝐫 \\ i+j=k}} a(i) b(j) }
		&= \abs[\bigg]{ ∑_{\substack{(i,j) ∈ [n]×[m] ∖ 𝐫 \\ i+j=k}} a(i) b(j) } \\
		& ≤ ∑_{\substack{(i,j) ∈ [n]×[m] ∖ 𝐫 \\ i+j=k}} \abs{a(i) b(j)} \\
		& ≤ (m+1) 2^{-wp-˘h(k)}.
	\end{align*}
	Combining the two inequalities above, we get what we need to prove.
\end{proof}

\begin{lemma}
	\label{d_equal_sum_r_e}
	For integers $i⁻ ≤ i⁺$, $j⁻ ≤ j⁺$,
	\[ d(i⁻, j⁺) + d(i⁺, j⁻) - d(i⁻, j⁻) - d(i⁺, j⁺)
	= ∑_{k ∈ ℤ} \val(R, k) ⋅ e(k), \]
	where $R = [i⁻, i⁺] × [j⁻, j⁺]$.
\end{lemma}

\begin{proof}
	When $i⁻ = i⁺$,
	the left hand side is $(d(i⁻, j⁺) - d(i⁺, j⁺)) + (d(i⁺, j⁻) - d(i⁻, j⁻)) = 0$,
	and the right hand side is $0$ since every $\val(R, k)$ is $0$,
	so the equality holds.
	The case $j⁻ = j⁺$ is similar.

	Next, consider the case $i⁺-i⁻ = j⁺-j⁻ = 1$.
	In this case,
	\[ \val(R, k) = \begin{cases} 1 \text{ if } i⁻+j⁻+1 = k, \\ 0 \text{ otherwise. }\end{cases} \]
	So the right hand side is $e_{i⁻+j⁻+1}$.
	The result follows from the definitions.

	If none of the cases above holds, we induct on $i⁺-i⁻ + j⁺-j⁻$.
	The case $i⁺-i⁻ + j⁺-j⁻ ≤ 2$ are already solved.
	When $i⁺-i⁻ + j⁺-j⁻ ≥ 3$, either $i⁺-i⁻ ≥ 2$ or $j⁺-j⁻ ≥ 2$.
	We handle the case $i⁺-i⁻ ≥ 2$, the other case is similar.

	Let $i^* = i⁻ + 1$, although any $i⁻ < i^* < i⁺$ would suffice.
	Let $R_1 = [i⁻, i^*] × [j⁻, j⁺]$ and $R_2 = [i^*, i⁺] × [j⁻, j⁺]$.
	We easily see that $R_1 ∪ R_2 = R$ and $R_1 ∩ R_2$ is a vertical segment with zero area.
	Therefore, for all $k$,
	\[ ∑_{k ∈ ℤ} \val(R, k) ⋅ e(k) = ∑_{k ∈ ℤ} \val(R_1, k) ⋅ e(k) + ∑_{k ∈ ℤ} \val(R_2, k) ⋅ e(k). \]
	Furthermore, we have
	\begin{multline*}
		d(i⁻, j⁺) + d(i⁺, j⁻) - d(i⁻, j⁻) - d(i⁺, j⁺) =
		\Bigl(d(i⁻, j⁺) + d(i^*, j⁻) \\ - d(i⁻, j⁻) - d(i^*, j⁺)\Bigr) +
		\Bigl(d(i^*, j⁺) + d(i⁺, j⁻) - d(i^*, j⁻) - d(i⁺, j⁺)\Bigr).
	\end{multline*}
	By the induction hypothesis, the result follows.
\end{proof}

\begin{lemma}
	\label{lemma_main}
	For real numbers $ε > 0$, $i⁻₁$, $i⁺₁$, $j⁻₁$, $j⁺₁$,
	$i⁻₂$, $i⁺₂$, $j⁻₂$, $j⁺₂$,
	$ℓ₁$, $ℓ₂$,
	such that 
	\begin{gather}
		ℓ₁ = i⁺₁-i⁻₁ = j⁺₁-j⁻₁ ≥ 0,  \notag \\
		ℓ₂ = i⁺₂-i⁻₂ = j⁺₂-j⁻₂ ≥ 0,  \notag \\
		i⁻₁ ≤ i⁻₂, \qquad j⁻₁ ≤ j⁻₂, \notag \\
		i⁺₁ ≤ i⁺₂, \qquad j⁺₁ ≤ j⁺₂, \notag \\
	(i⁺₁+j⁺₁) - (i⁻₂+j⁻₂) ≥ ε ⋅ \min(ℓ₁, ℓ₂),  \label{overlap_amount}
	\end{gather}
	then, let
	\begin{align*}
		R₁ &= [i⁻₁, i⁺₁] × [j⁻₁, j⁺₁], \\
		R₂ &= [i⁻₂, i⁺₂] × [j⁻₂, j⁺₂], \\
		R₃ &= [i⁻₁, i⁺₂] × [j⁻₁, j⁺₂],
	\end{align*}
	we have, for every $k ∈ ℝ$,
	\[ \val(R₃, k) ≤ \frac{2}{ε} ⋅ \bigl(\val(R₁, k) + \val(R₂, k)\bigr). \]
\end{lemma}

Note that $R_1$ and $R_2$ are squares.
See \Cref{fig_main_lemma} for an illustration.

	\begin{figure}
		\centering
		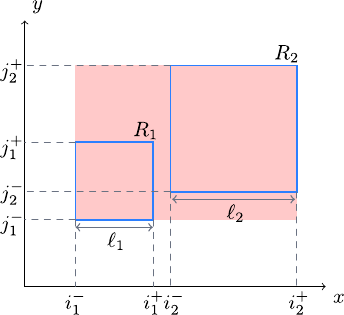
		\caption{
	Illustration for \Cref{lemma_main}.
	$R_3$ (shaded light red) is the smallest rectangle containing both $R_1$ and $R_2$.
}
		\label{fig_main_lemma}
	\end{figure}

While the lemma is intuitively obvious, trying to formalize the intuitive proof is extremely tedious.

\begin{proof}
	First, we handle the case where $R_3$ is a square.
	In this case, $i⁻₁ + j⁺₂ = i⁺₂ + j⁻₁$. 

	For a square $R = [i⁻, i⁺] ×[j⁻, j⁺]$ with $ℓ = i⁺-i⁻ = j⁺-j⁻$, we have $\val(R, k) = \max(0, ℓ-|k-i⁺-j⁻|)
	= \max(0, \min(k-i⁻-j⁻, i⁺+j⁺-k))$.

	We set, for both $r ∈ \{ 1, 2 \}$,
	\[
		k⁻ᵣ = i⁻ᵣ+j⁻ᵣ, \quad 
		k^*ᵣ = i⁻ᵣ+j⁺ᵣ = i⁺ᵣ+j⁻ᵣ, \quad 
		k⁺ᵣ = i⁺ᵣ+j⁺ᵣ.
	\]

	Without loss of generality, we assume $ℓ₁ ≤ ℓ₂$. The other case is handled symmetrically.

	We have $i⁻₁ ≤ i⁻₂$ and $j⁻₁ ≤ j⁻₂$. Combining with 
	\Cref{overlap_amount}, we get
	\[ 
		2 ℓ₁ = k⁺₁-k⁻₂ ≥ ε \min(ℓ₁, ℓ₂) = ε ℓ₁,
	\]
	so $ε ≤ 2$.

	When $k ∉ [k⁻₁, k⁺₂]$, $\val(R_r, k) = 0$ for all $r ∈ \{ 1, 2, 3 \}$,
	so it suffices to assume $k ∈ [k⁻₁, k⁺₂]$.

	Consider the following cases.

	\paragraph{Case 1} $k⁻₁ ≤ k ≤ k^*₁$. Then $\val(R₁, k) = \val(R₃, k) = k-k⁻₁$, combined with $ε ≤ 2$, the result holds.

	\paragraph{Case 2} $k^*₂ ≤ k ≤ k⁺₂$. Then $\val(R₂, k) = \val(R₃, k) = k⁺₂-k$, the result holds similar to above.

	\paragraph{Case 3} $k^*₁ ≤ k ≤ k⁻₂$. Then $\val(R₁, k) ≥ ε ℓ₁$, $\val(R₂, k) = 0$,
	and $\val(R₃, k) ≤ k-k⁻₁ ≤ 2 ℓ₁$.

	\paragraph{Case 4} $k⁻₂ ≤ k ≤ k⁺₁$. Then $\val(R₁, k) + \val(R₂, k) = ε ℓ₁$,
	and $\val(R₃, k) ≤ k-k⁻₁ ≤ 2 ℓ₁$.

	\paragraph{Case 5} $k⁺₁ ≤ k ≤ k^*₂$. Then $\val(R₁, k) = 0$, $\val(R₂, k) = k-k⁻₂$,
	and $\val(R₃, k) ≤ k-k⁻₁$.
	Equivalently,
	$\val(R₂, k) = \val(R₂, k⁺₁) + (k-k⁺₁)$, and
	$\val(R₃, k) ≤ \val(R₃, k⁺₁) + (k-k⁺₁)$.
	As in case 4, $\val(R₃, k⁺₁) ≤ \frac{2}{ε} ⋅ \val(R₂, k⁺₁)$,
	and since $ε ≤ 2$, $k-k⁺₁ ≤ \frac{2}{ε} (k-k⁺₁)$.

	To handle the non-square $R₃$ case, we reduce the problem to the square $R₃$ case.
	Define:
	\begin{align*}
		δ &= \frac{(i⁻₂-j⁻₂) - (i⁻₁-j⁻₁)}{2}, \\
		i⁻₄ &= i⁻₂ - δ, \\
		i⁺₄ &= i⁺₂ - δ, \\
		j⁺₄ &= j⁺₂ + δ, \\
		j⁻₄ &= j⁻₂ + δ, \\
		R₄ &= [i⁻₄, i⁺₄] × [j⁻₄, j⁺₄], \\
		R₅ &= [i⁻₁, i⁺₄] × [j⁻₁, j⁺₄].
	\end{align*}
	Intuitively, $R₅$ is the smallest rectangle containing both $R₁$ and $R₄$,
	and $R₄$ is obtained by moving $R₂$ diagonally such that $R₅$
	is a square. See \Cref{fig_main_lemma_move_to_square} for an illustration.
	
	\begin{figure}
		\centering
		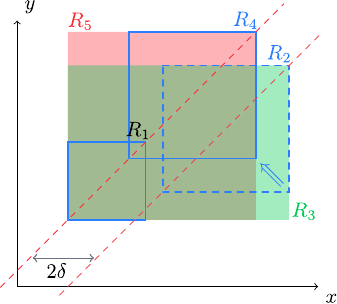
		\caption{
		Illustration for the reduction from the non-square $R₃$ to the square $R₃$ case
		in the proof of \Cref{lemma_main}.
		Note that $δ$ is a signed distance, it is positive in this example.
	}
		\label{fig_main_lemma_move_to_square}
	\end{figure}

	It can be checked that $R₅$ is a square.
	By the other case of the proof, we have, for every $k ∈ ℝ$,
	\[ \val(R₅, k) ≤ \frac{2}{ε} ⋅ \bigl(\val(R₁, k) + \val(R₄, k)\bigr). \]
	It can be checked that $\val(R₄, k) = \val(R₂, k)$.
	So, it suffices for us to prove, for every $k ∈ ℝ$,
	\[ \val(R₃, k) ≤ \val(R₅, k). \]

	If $δ = 0$, $R₃ = R₅$, the statement is obvious.

	We assume $δ > 0$. The other case, $δ<0$, is similar.
	See \Cref{fig_main_lemma_move_to_square_2} for an illustration.

	\begin{figure}
		\centering
		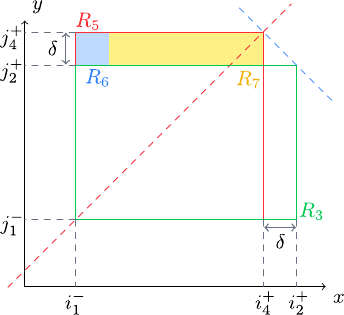
		\caption{
		Illustration for the proof of $\val(R₃, k) ≤ \val(R₅, k)$.
	}
		\label{fig_main_lemma_move_to_square_2}
	\end{figure}

	We have:
	\begin{align*}
		R₃ ∩ R₅ &= [i⁻₁, i⁺₄] × [j⁻₁, j⁺₂], \\
		R₃ ∖ R₅ &= [i⁺₄, i⁺₂] × [j⁻₁, j⁺₂], \\
		R₅ ∖ R₃ &= [i⁻₁, i⁺₄] × [j⁺₂, j⁺₄].
	\end{align*}
	So, it suffices to prove $\val(R₃ ∖ R₅, k) ≤ \val(R₅ ∖ R₃, k)$.
	Define:
	\begin{align*}
		R₆ &= [i⁻₁, i⁻₁ + δ] × [j⁺₂, j⁺₄], \\
		R₇ &= [i⁻₁ + δ, i⁺₄] × [j⁺₂, j⁺₄].
	\end{align*}
	Then it can be seen that $\val(R₃ ∖ R₅, k) = \val(R₇, k)$,
	and $\val(R₅ ∖ R₃, k) - \val(R₇, k) = \val(R₆, k) ≥ 0$.
\end{proof}

\begin{lemma}
	\label{lemma_almost_square}
	With the same assumptions as \Cref{lemma_main}. Then let $δ_i = i⁺₂-i⁻₁$, $δ_j = j⁺₂-j⁻₁$,
	for $λ = 3-ε$,
	then $δ_i ≤ λ ⋅ δ_j$ and $δ_j ≤ λ ⋅ δ_i$.
\end{lemma}
\begin{proof}
	Without loss of generality, assume $ℓ₁ ≤ ℓ₂$ and $δ_j ≤ δ_i$.
	We need to prove $δ_i ≤ λ δ_j$.

	By \Cref{overlap_amount}, we have $(i⁺₁+j⁺₁)-(i⁻₂+j⁻₂) ≥ ε ℓ₁$.
	Combining with $j⁻₂ ≥ j⁻₁$, we have
	$(i⁺₁+j⁺₁)-i⁻₂ ≥ ε ℓ₁ + j⁻₁$,
	or equivalently $i⁺₁-i⁻₂ ≥ (ε-1) ℓ₁$.
	Since $i⁺₁=i⁻₁+ℓ₁$ and $i⁻₂=i⁺₂-ℓ₂$, we have $i⁻₁-i⁺₂+ℓ₂ ≥ (ε-2) ℓ₁$,
	or equivalently, $δ_i ≤ ℓ₂+(2-ε) ℓ₁ ≤ (3-ε) ℓ₂ = λ ⋅ ℓ₂$.
	Combining with $j⁺₂ = j⁻₂ + ℓ₂ ≥ j⁻₁ + ℓ₂ ⟹ δ_j = j⁺₂-j⁻₁ ≥ ℓ₂$,
	we're done.
\end{proof}

\begin{lemma}
	\label{d_sum_bound}
	With the same notation and hypotheses as in \Cref{lemma_main}.
	Furthermore, assume for every $(i, j) ∈ (R_1 ∪ R_2) ∩ ℤ^2$, $d(i, j) ≤ δ$ for some constant $δ ≥ 0$.
	Then, $d(i⁻₁, j⁺₂) + d(i⁺₂, j⁻₁) ≤ 8δ/ε$.
\end{lemma}
\begin{proof}
	By \Cref{lemma_main},
	for every $k$,
	\[ \val(R₃, k) ≤ \frac{2}{ε} ⋅ \bigl(\val(R₁, k) + \val(R₂, k)\bigr). \]
	Therefore,
	\[ ∑_{k ∈ ℤ} \val(R₃, k) e(k) ≤ ∑_{k ∈ ℤ} \frac{2}{ε} ⋅ \bigl(\val(R₁, k) + \val(R₂, k)\bigr) e(k). \]
	By \Cref{d_equal_sum_r_e}, we get
	\begin{multline*}
		d(i⁻₁, j⁺₂) + d(i⁺₂, j⁻₁) - d(i⁻₁, j⁻₁) - d(i⁺₂, j⁺₂)
		≤ \frac{2}{ε} ⋅ \Bigl(
			d(i⁻₁, j⁺₁) + d(i⁺₁, j⁻₁) \\ - d(i⁻₁, j⁻₁) - d(i⁺₁, j⁺₁) +
			d(i⁻₂, j⁺₂) + d(i⁺₂, j⁻₂) - d(i⁻₂, j⁻₂) - d(i⁺₂, j⁺₂)
		\Bigr).
	\end{multline*}
	As shown in the proof of \Cref{lemma_main}, $ε ≤ 2$.
	Therefore,
	\begin{multline*}
		d(i⁻₁, j⁺₂) + d(i⁺₂, j⁻₁) 
		≤ \frac{2}{ε} ⋅ \Bigl(
			d(i⁻₁, j⁺₁) + d(i⁺₁, j⁻₁) \\ - d(i⁺₁, j⁺₁) +
			d(i⁻₂, j⁺₂) + d(i⁺₂, j⁻₂) - d(i⁻₂, j⁻₂)
		\Bigr).
	\end{multline*}
	Since each $d(⋅)$ term on the right hand side is $≥ 0$ and $≤ δ$,
	the result follows.
\end{proof}

\begin{lemma}
	\label{max_d_on_corner}
	Let $R = [i⁻ ‥ i⁺] ×[j⁻ ‥ j⁺]$ be a rectangle.
	Then $\max d(R) ∈ \{ d(i⁻, j⁺), d(i⁺, j⁻) \}$.
\end{lemma}
\begin{proof}
	Follows immediately from \cite[Lemma~18, Monotone Paths]{bringmann_2023_arxiv}.
\end{proof}

\begin{lemma}
	\label{lemma_max_d_le}
	With the same notation and hypotheses as in \Cref{d_sum_bound}.
	Then $\max_{(i, j) ∈ R₃} d(i, j) ≤ 8δ/ε$.
\end{lemma}
\begin{proof}
	Since $d(i⁻₁, j⁺₂) ≥ 0$ and $d(i⁺₂, j⁻₁) ≥ 0$,
	from the conclusion of \Cref{d_sum_bound}, we know
	$\max(d(i⁻₁, j⁺₂), d(i⁺₂, j⁻₁)) ≤ 8δ/ε$.
	Applying \Cref{max_d_on_corner}, the result follows.
\end{proof}

\begin{proposition}
	\label{sum_side_rectangles_le}
	For the sequence $‾k$ constructed in \Cref{sec_main_algorithm},
	we have
	\[
		∑_{q ∈ [t-1]} \Bigl(i⁺_{‾k(q+1)} + j⁺_{‾k(q+1)} - i⁻_{‾k(q)} - j⁻_{‾k(q)}\Bigr)
		< \frac{3κ-8}{κ-4} (n+m).
	\]
\end{proposition}
Recall $i⁺_k$ and $i⁻_k$ was defined in \Cref{def_ij_plus_minus_k}.
Recall we assumed $κ>4$ in \Cref{kappa_ge_4}.
We believe that any $κ>1$ would work (possibly with a larger multiplicative constant on the right hand side), but the proof would need to be more involved.
\begin{proof}
	Fix a particular $q ∈ [t-2]$.
	Since $q+1 ≠ t$, we know $‾k(q+1) ≤ n+m$.
	By the construction of $‾k$,
	$‾k(q+1)$ is the largest integer
	such that $‾k(q+1) ≤ n+m+1$ and
	\[ \max\Bigl(d\bigl(i⁻_{‾k(q)},j⁺_{‾k(q+1)-1}\bigr),d\bigl(i⁺_{‾k(q+1)-1},j⁻_{‾k(q)}\bigr)\Bigr) ≤ κ ⋅ w ⋅ p. \]
	Combined with $‾k(q+1) ≤ n+m$, we conclude that
	\[ \max\Bigl(d\bigl(i⁻_{‾k(q)},j⁺_{‾k(q+1)}\bigr),d\bigl(i⁺_{‾k(q+1)},j⁻_{‾k(q)}\bigr)\Bigr) > κ ⋅ w ⋅ p. \]
	Then, apply \Cref{lemma_max_d_le} with
	\begin{align*}
		ε &= \frac{8}{κ}, \\
		δ &= w ⋅ p, \\
		R_1 &= [i⁻_{‾k(q)}, i⁺_{‾k(q)}] × [j⁻_{‾k(q)}, j⁺_{‾k(q)}], \\
		R_2 &= [i⁻_{‾k(q+1)}, i⁺_{‾k(q+1)}] × [j⁻_{‾k(q+1)}, j⁺_{‾k(q+1)}], \\
		R_3 &= [i⁻_{‾k(q+1)}, i⁺_{‾k(q+1)}] × [j⁻_{‾k(q+1)}, j⁺_{‾k(q+1)}],
	\end{align*}
	we see that the conclusion \emph{does not} hold, so at least one of the hypotheses fail.
	In this case, the only hypothesis that can fail is \Cref{overlap_amount}.
	So, we conclude
	\[
		\bigl(i⁺_{‾k(q)} + j⁺_{‾k(q)}\bigr) - \bigl(i⁻_{‾k(q+1)} + j⁻_{‾k(q+1)}\bigr)
		< ε ⋅ \min\bigl( i⁺_{‾k(q)} - i⁻_{‾k(q)}, i⁺_{‾k(q+1)} - i⁻_{‾k(q+1)} \bigr).
	\]
	For brevity, similar to the proof of \Cref{lemma_main}, we define
	\begin{align*}
		k⁻₁  &= i⁻_{‾k(q)} + j⁻_{‾k(q)}, & k⁻₂  &= i⁻_{‾k(q+1)} + j⁻_{‾k(q+1)}, \\
		k^*₁ &= i⁻_{‾k(q)} + j⁺_{‾k(q)}, & k^*₂ &= i⁻_{‾k(q+1)} + j⁺_{‾k(q+1)}, \\
		k⁺₁  &= i⁺_{‾k(q)} + j⁺_{‾k(q)}, & k⁺₂  &= i⁺_{‾k(q+1)} + j⁺_{‾k(q+1)}, \\
		ℓ₁   &= i⁺_{‾k(q)} - i⁻_{‾k(q)}, & ℓ₂   &= i⁺_{‾k(q+1)} - i⁻_{‾k(q+1)}.
	\end{align*}
	So, we have
	\begin{align*}
		k⁺₁ - k⁻₂ &< ε \min(ℓ₁, ℓ₂) \\
		⟺ k^*₁ + ℓ₁ - k^*₂ + ℓ₂ &< ε \min(ℓ₁, ℓ₂) \\
		⟺ k^*₂ - k^*₁ &> ℓ₁ + ℓ₂ - ε \min(ℓ₁, ℓ₂) \\
		⟹ k^*₂ - k^*₁ &> (1 - \tfrac{ε}{2}) (ℓ₁ + ℓ₂).
	\end{align*}
	With some rearrangement, combined with $κ > 4 ⟹ ε < 2$, we get
	\[
		k^*₂ - k^*₁ > \frac{2-ε}{4-ε} ⋅ (k⁺₂ - k⁻₁).
	\]
	Substituting in the definitions, we have, for all $q ∈ [t-2]$,
	\[
		\bigl(i⁺_{‾k(q+1)} + j⁻_{‾k(q+1)}\bigr) - \bigl(i⁺_{‾k(q)} + j⁻_{‾k(q)}\bigr)
		> \frac{2-ε}{4-ε} ⋅ \bigl(i⁺_{‾k(q+1)} + j⁺_{‾k(q+1)} - i⁻_{‾k(q)} - j⁻_{‾k(q)}\bigr).
	\]
	Summing over all $q ∈ [t-2]$ and telescope the sum on the left hand side, we get
	\[
		i⁺_{‾k(t-1)} + j⁻_{‾k(t-1)}
		> \frac{2-ε}{4-ε} ⋅ ∑_{q ∈ [t-2]} \bigl(i⁺_{‾k(q+1)} + j⁺_{‾k(q+1)} - i⁻_{‾k(q)} - j⁻_{‾k(q)}\bigr).
	\]
	Note that $‾k(0) = 0$, so $i⁺_{‾k(0)} = j⁻_{‾k(0)} = 0$.
	Combining with all $i⁻_k$ and $i⁺_k$ are in $[n]$, and all $j⁻_k$ and $j⁺_k$ are in $[m]$, we get the result.
\end{proof}

\section{Time complexity of the main algorithm}

\begin{lemma}
	Consider the algorithm $\RescaleRound{⋅}$ in \Cref{algorithm:linear_scaling}.
	Suppose the input satisfies $\abs{a(i)/2^{σ ⋅ i+α}} ∈ 2^{O(wp)}$.
	Then the algorithm can be implemented in time $O\bigl((i⁺-i⁻+1) ⋅ p \log p + (\log^2 w + p \log p) \log(wp)\bigr)$.
\end{lemma}
Note that the naive algorithm takes $O\bigl((i⁺-i⁻+1) ⋅ (\log^2 w + p \log p) \log(wp)\bigr)$.
\begin{proof}
	First, compute a $O(p)$-word approximation to $2^{-σ}$ and $2^{-α-σ ⋅ i⁻}$.
	Then, compute values of $2^{-σ ⋅ i-α}$ for all $i ∈[i⁻ ‥ i⁺]$ by starting with $2^{-α}$ and repeatedly multiply $2^{-σ}$.
	The rest can be implemented with a naive algorithm.

	Note that the precision of the intermediate values only need to be sufficient for
	\Cref{error_requirement_rescale_round} to hold. As explained in \Cref{remark_loose_round_requirement},
	this is not a very strict requirement.
	We see that 
	the relative accumulated error is bounded by $2^{-O(wp)} ⋅ (i⁺+i⁻+1) ⋅ O(1)$.
	Typically, $i⁺+i⁻+1 ≪ 2^w$, so the relative accumulated error is bounded by $2^{-O(wp)}$.
	By choosing the constant factor in the exponent appropriately,
	this can satisfy \Cref{error_requirement_rescale_round}.
\end{proof}
In practice, one can use a binary splitting algorithm, so that the accumulated error
is only $2^{-w(p+1)} ⋅ O(i⁺+i⁻+1) ⋅ O(1)$, while requiring the same number of multiplications.

The following variant may be of interest in real implementations.

\begin{lemma}
	\label{rescale_round_time_complexity_variant}
	Consider the algorithm $\RescaleRound{⋅}$ in \Cref{algorithm:linear_scaling}.
	Suppose the input satisfies $\abs{a(i)/2^{σ ⋅ i+α}} ∈ 2^{O(wp)}$,
	$σ$ has denominator $2^k$ for a positive integer $k ∈ O(i⁺-i⁻+1)$,
	and any $Ω(w)$-bit square root approximation can be computed in $O(1)$ time.
	Then the algorithm can be implemented in time $O\bigl((i⁺-i⁻+1) ⋅ p \log p\bigr)$.
\end{lemma}
\begin{proof}
	Recall a square root can be computed to $O(p)$ words of precision 
	in time $O(p \log p)$ from a $Ω(w)$-bit approximation.
	The $p \log p \log k$ factor gets absorbed into the
	$(i⁺-i⁻+1) ⋅ p \log p$ factor.
	The rest follows immediately.
\end{proof}
An easy way to get a
$Ω(w)$-bit square root approximation in $O(1)$ time
is to prepare a lookup table with $w^{Ω(1)}$ word-sized entries.
This is a common method used by hardware in practice.  
If $n ∈ w^{Ω(1)}$, then preparing the lookup table takes less than linear time,
which can be considered negligible.

\begin{lemma}
	\label{assume_rescaled_convolution}
	Consider $\RescaledConvolution{⋅}$ procedure in \Cref{algorithm:linear_scaling}.
	Suppose the input satisfies $\abs{a(i)/2^{σ ⋅ i+α}} ∈ 2^{O(wp)}$,
	$\abs{b(j)/2^{σ ⋅ j+β}} ∈ 2^{O(wp)}$ for all $(i, j) ∈ 𝐫$,
	$σ$ has denominator $2^k$ for a positive integer $k ∈ O(i⁺-i⁻+1)$,
	and any $Ω(w)$-bit square root approximation can be computed in $O(1)$ time.
	Then the time complexity of the algorithm is $O\bigl(‾n p \log ( ‾n p )\bigr)$,
	where $‾n = i⁺+j⁺-i⁻-j⁻+1$.
\end{lemma}
\begin{proof}
	Using \Cref{rescale_round_time_complexity_variant}, we can compute $‾a$ and $‾b$ in the claimed bound.
	We know from the introduction that, by Kronecker substitution,
	two polynomials with $O(p)$-word integer coefficients and
	total degree $n$ can be multiplied in time $O(np \log (np))$.
	Finally, $\Unscale{⋅}$ can be implemented with the same time complexity as $\RescaleRound{⋅}$
	using the same idea as in \Cref{rescale_round_time_complexity_variant}.
\end{proof}

\begin{theorem}
	\label{main_algorithm_time_complexity}
	$\textproc{ApproximateConvexConvolution}(⋅)$ in \Cref{algorithm:block_division} can be implemented in $O(np \log (np))$ time.
\end{theorem}

\begin{proof}
	Clearly the sequence $‾k$ can be computed in $O(n+m)$ time.
	Since we assume $m ≤ n$, this is $O(n)$.

	Apart from the final return statement, the algorithm consists of $t$ iterations, one for each value $q ∈ [t-1]$.
	We fix any $q ∈ [t-1]$.
	Following the algorithm code, we let $k⁻ = ‾k(q)$ and $k⁺ = ‾k(q+1)-1$.

	Since the rational numbers $σ$, $α$, $β$ computed at the beginning of the procedure $\AutoRescaledConvolution{⋅}$
	fit in $O(1)$ words, most of the time is spent in the $\RescaledConvolution{⋅}$ subroutine.
	Let $𝐫 = [i⁻_{k⁻}‥ i⁺_{k⁺}]× [j⁻_{k⁻}‥ j⁺_{k⁺}]$.
	By \Cref{max_d_corner_le}, which bounds the value of $d(⋅)$ on two corners of $𝐫$ by $κ⋅w⋅p$,
	and \Cref{max_d_on_corner}, which states that in a rectangle, the maximum value of $d(⋅)$
	can only be on certain corners, for every $(i, j) ∈ 𝐫$, $d(i, j) ≤ κ⋅w⋅p$.

	Set $δ = κwp$.
	\Cref{prop_linear_scaling_on_rectangle} states that, if $d(i, j) ≤ δ$ for all $(i, j) ∈ 𝐫$,
	and some other hypotheses are satisfied, then $˘f$ and $˘g$ are nearly linear on that rectangle.
	We apply this proposition on this rectangle $𝐫$ with $λ = 3$.
	The assumptions are satisfied because of \Cref{lemma_almost_square}, which states that the longer side of $𝐫$ is not much longer than the shorter side, and
	$d(i, j) ≤ κ⋅w⋅p$ for every $(i, j) ∈ 𝐫$, as proved above.
	With the $σ$, $α$ and $β$ computed in \Cref{algorithm:auto_rescaled_convolution},
	for every $(i, j) ∈ 𝐫$, $-˘f(i, j) - (σi-α) ≤ κwp$.
	We can apply \Cref{assume_rescaled_convolution} to conclude that,
	for this iteration, the invocation of $\AutoRescaledConvolution{⋅}$ takes time 
	$O(‾n κ p \log (‾n κ p))$, where $‾n = i⁺_{k⁺} - i⁻_{k⁻} + j⁺_{k⁺} - j⁻_{k⁻} + 1$.

	Applying \Cref{sum_side_rectangles_le}, we conclude the total time complexity of all
	$t$ iterations, dropping the constant factor $κ$, is $O\bigl((n+m) p \log ((n+m) p)\bigr)$.
	Since $n ≥ m$ by definition, this is $O(np \log(np))$ as needed.
\end{proof}

\section{Discussion}
\label{sec_discussion}

\subsection{Motivation for the definition of relative Newton error}

Earlier, we remarked that our algorithm gives good result when the
product polynomial satisfies \Cref{optimistic_c_bound}, in other words,
the coefficient sequence of the product is near-log-convex.
Which raises the question: if this does not hold, how good is our algorithm?

We rigorously proved above that the algorithm will produce a polynomial with small relative Newton error.
But this only raise another question: Why do we choose to define relative Newton error the way we did?
Part of the reason is because the error of the algorithm naturally gives this bound.
But there are other justifications.

The best we can hope for is for the algorithm to have small \emph{relative error},
say, for each $k ∈[n+m]$, $|˜c(k)-c(k)| ≤ 2^{-wp} |c(k)| ⋅ O(1)$.
This definition is too strict: consider the two factor polynomials being $x+1$ and $x-1$ respectively,
then their product is $x^2-1$, with $c(1)=0$.
If the first factor is slightly modified to $x+1+2^{-wp-1}$, $c(1)$ would become $2^{-wp-1}$,
and the relative error would be $|2^{-wp-1}-0|/0 = ∞$.
If each coefficient of the input is given as a $p$-word floating point, it is reasonable to assume that
the input is only known to approximately $p$ words of precision, and introducing a relative error of $O(2^{-wp})$
to an input coefficient shouldn't lead to significant error in the output.
This is why, in classical analysis of numerical stability of algorithms,
the definitions of both forward stability and backward stability allow for variation in the input,
see \cite[Section~1.5]{higham2002accuracy} and \cite[Lecture~14]{trefethen2022numerical}.

Slightly less ambitiously, we may hope for the result $˜c$ to satisfy:
for some $‾a \colon[n] → ℝ$ and $‾b \colon[m] → ℝ$ such that $|‾a-a| ≤ ε |a|$ and $|‾b-b| ≤ ε |b|$,
let $‾c \colon [n+m] → ℝ$ be the convolution of $‾a$ and $‾b$,
then $|˜c-‾c| ≤ ε |‾c|$ holds.\footnote{We recall arithmetic and comparison operators are performed termwise:
$|‾a-a| ≤ ε |a|$ means
$|‾a(i)-a(i)| ≤ ε |a(i)|$ for all $i ∈ [n]$.}
This is achieved with $ε = O(1) ⋅ n ⋅ 2^{-wp}$ by the naive multiplication algorithm, where the intermediate values
in the computation are stored as $p$-word floating-point numbers.  
Unfortunately, the naive multiplication algorithm takes time quadratic in the degree.
For any $ε ∈ O(1)$, 
the min-plus convolution problem
can be reduced to
this problem, computing a polynomial product with small forward coefficientwise relative error,
with only a constant factor overhead.
As discussed in \Cref{sec_related_problems}, the min-plus convolution problem
is believed to be unsolvable in time faster than $n^{2-o(1)}$,
so we don't expect to be able to solve this problem either.

This is why we take the convex hull:
our bound is almost equivalent to $|‾a-a| ≤ ε 2^{-˘f}$ and $|‾b-b| ≤ ε 2^{-˘g}$.

Can the algorithm be backwards stable? In other words, can we hope for
the result $˜c$ to be \emph{exactly} the convolution of $‾a$ and $‾b$,
for $|‾a-a|$ and $|‾b-b|$ both small?
A simple example shows this is not the case:
when $a$ and $b$ are both the list of coefficients of $x+1$,
then $c$ is the list of coefficients of $(x+1)^2$.
If the algorithm makes a small error in its calculation, so that $c = (x+1)^2+ε$,
then we necessarily have $‾a = x+1+√{-ε}$.
For $0 < ε ≪ 1$, $√ε$ is much larger than $ε$.

How good is the relative Newton error? One way to argue for it is the following, motivated by \cite[Section~9]{stablemult}.

Fix any point $z ∈ ℂ$. We consider the evaluation $∑ a(i) z^i$ of the polynomials $a$ at the point $z$.
If for some $i^* ∈ [n]$, the coefficient $a(i^*)$ is slightly perturbed to $a(i^*) ⋅ (1+ε)$
for some $ε ∈ ℂ$, $0<|ε|≤ 1$,
then the evaluation result is perturbed to $∑ a(i) z^i + ε a(i^*) z^{i^*}$.
So, it's sensible that we define the relative condition number
\[ \RelCond(a, z) = \frac{\max_{i ∈ [n]} \abs{a(i) z^i}}{\abs[\big]{∑_{i ∈ [n]} a(i) z^i}}. \]
If this value is small, we can say that we know with low relative error the evaluation of $a$ at $z$.

We can also define
\[ \RelNewtonCond(a, z) = \frac{\max_{i ∈ [n]} \abs{2^{-˘f(i)} z^i}}{\abs[\big]{∑_{i ∈ [n]} a(i) z^i}}. \]

Then we have
\begin{proposition}
	$\RelNewtonCond(a, z) = \RelCond(a, z)$.
\end{proposition}
\begin{proof}
	The denominator is the same, so it suffices to prove
	\[ \max_{i ∈ [n]} \, \abs{2^{-˘f(i)} z^i} = \max_{i ∈ [n]} \, \abs{a(i) z^i} . \]
	Since $˘f(i) ≤ f(i)$ for every $i$, 
	$2^{-˘f(i)} ≥ a(i)$.
	So we get
	\[ \max_{i ∈ [n]} \, \abs{2^{-˘f(i)} z^i} ≥ \max_{i ∈ [n]} \, \abs{a(i) z^i}. \]
	It remains to prove, for every $i$,
	\[ \abs{2^{-˘f(i)} z^i} ≤ \max_{i ∈ [n]} \, \abs{a(i) z^i}. \]
	Fix a particular $i$. The case $˘f(i) = f(i)$ is trivial, so we assume $˘f(i) < f(i)$.
	By the definition of the convex hull,
	there exists $i⁻ < i < i⁺$
	such that $˘f(i⁻) = f(i⁻)$, $˘f(i⁺) = f(i⁺)$,
	and
	\[ ˘f(i) = \frac{˘f(i⁻) ⋅ (i⁺-i) + ˘f(i⁺) ⋅ (i-i⁻)}{i⁺-i⁻}.  \]
	In words, the point $(i, ˘f(i))$ is on the line segment connecting
	$(i⁻, ˘f(i⁻))$ and $(i⁺, ˘f(i⁺))$.
	Let $λ = \frac{i⁺-i}{i⁺-i⁻}$, then $0 < λ < 1$. Easy algebraic manipulations show that
\[ \abs[\big]{2^{-˘f(i)} z^i} = \abs[\big]{2^{-˘f(i⁻)} z^{i⁻}}^λ ⋅ \abs[\big]{2^{-˘f(i⁺)} z^{i⁺}}^{1-λ} = \abs[\big]{a(i⁻) z^{i⁻}}^λ ⋅ \abs[\big]{a(i⁺) z^{i⁺}}^{1-λ}.  \]
	Therefore, $\abs{2^{-˘f(i)} z^i} ≤ \max(\abs{a(i⁻) z^{i⁻}}, \abs{a(i⁺) z^{i⁺}})$, and we're done.
\end{proof}

Is $\RelCond(a, z)$ small for reasonable polynomials $a$?
Unfortunately not: as $z$ approaches a root of $a$, the denominator approaches zero, so $\RelCond(a, z)$ approaches infinity.

\subsection{A difference from previous work in a definition}
\label{sec_discuss_def_relative_Newton_error}

With notation as in \Cref{def_relative_Newton_error}.
In \cite[Section~5]{stablemult}, the author defines
the relative Newton error of $˜c$ with respect to $c$ to be
\[ \max_{k ∈ [n+m]} \frac{|˜c(k)-c(k)|}{2^{-˘h'(k)}}, \]
where $˘h'$ is the lower convex hull of $- \log_2 |c|$.
This is slightly different from our definitions.

We will see that $˘h'$ gives \emph{better} error bound than $˘h$, up to a factor polynomial in the degree.
\begin{proposition}
	For each $k ∈[n+m]$, $2^{-˘h'(k)} ≤ (n+1) 2^{-˘h(k)}$.
\end{proposition}
Recall $n$ is the degree of $a$, $m$ is the degree of $b$, and $n ≥ m$.
\begin{proof}
	Fix any $k ∈[n+m]$.
	For any $(i,j)∈[n] ×[m]$ such that $i+j=k$, we have $˘h(k) ≤ ˘f(i)+˘g(j)$.
	So $2^{-˘h(k)}≥ 2^{-˘f(i)} 2^{-˘g(i)} ≥ a(i) b(j)$.
	Since $c(k) = ∑_{i+j=k} a(i) b(j)$, and there are at most $n+1$ terms in the sum in the right hand side,
	$\abs{c(k)} ≤ (n+1) \max_{i+j=k} \abs{a(i) b(j)} ≤ (n+1) 2^{-˘h(k)}$.

	Since this holds for any $k$, we have $- \log_2 |c| ≥ ˘h - \log_2 (n+1)$.
	Since $˘h$ is convex, so is $˘h - \log_2 (n+1)$.
	Since $- \log_2 |c|$ is is pointwise greater than or equal to the convex function $˘h - \log_2 (n+1)$,
	its lower convex hull is also pointwise greater than $˘h - \log_2 (n+1)$. The result follows.
\end{proof}
Sometimes, it is better by an exponential factor.
\begin{example}
	Let $n$ be a large even positive integer.
	Let $∑_i a(i) x^i = (x+1)^n$ and $∑_j b(j) x^j = (x-1)^n$.  Then $∑_k c(k) x^k = (x^2-1)^n$.
	It is easy to verify that $\max_i \abs{a(i)}= \max_j \abs{b(j)}=\max_k \abs{c(k)} = \binom{n}{n/2}$.
\end{example}

Does this mean the algorithm in \cite{stablemult} can be exponentially better than our algorithm?
On close scrutiny, it turns out that there is an error in their analysis.
In particular, \cite[Equation~2]{stablemult} states (after translating into the notation of this article)
\[
	‖c‖ ≥ \frac{1}{2n} ‖a‖ \, ‖b‖
\]
where $‖a‖ = \max_i |a(i)|$ denotes the $L^1$-norm of $a$.
The error ultimately comes from the claim without proof that
\[ ‖ \FFT_ω (a) ‖ \, ‖ \FFT_ω (b) ‖ = ‖ \FFT_ω (c) ‖, \]
where $ω$ is a $2n$-th root of unity, and $\FFT_ω$ denotes the FFT transform with respect to $ω$:
for polynomial $a = ∑_i a(i) x^i$,
$\FFT_ω (a)$ is a sequence of $2n$ complex numbers defined by
\[ \FFT_ω (a) \colon [2n-1] → ℂ, \quad \FFT_ω (a)(i) = ∑_i a(i) ω^i. \]
It is clear that 
$\FFT_ω (c)$ is the pointwise product of $\FFT_ω (a)$ and $\FFT_ω (b)$,
but this only mean $\Abs{\FFT_ω (a)} \, \Abs{\FFT_ω (b)} ≥ \Abs{\FFT_ω (c)}$.
The necessary and sufficient condition for equality to hold
is \[ \argmax_i \abs{\FFT_ω (a)(i)} ∩ \argmax_i \abs{\FFT_ω (b)(i)} ≠ ∅, \]
where $\argmax_i f(i)$ is the set of values $i$ such that $f(i)$ attains its maximum.

After fixing the error,
the error bound becomes mostly the same as ours.
Therefore, we also modify the definition of relative Newton error in order to
correctly state the error of our algorithm.

\hbadness=10000

\bibliographystyle{alpha}
\bibliography{note}

@book{higham2002accuracy,
title = "Accuracy and Stability of Numerical Algorithms",
author = "Nicholas Higham",
year = "2002",
month = jan,
day = "1",
language = "English",
isbn = "0-89871-521-0",
publisher = "Society for Industrial and Applied Mathematics",
address = "United States",
}

@book{trefethen2022numerical,
  title={Numerical linear algebra},
  author={Trefethen, Lloyd N and Bau, David},
  year={2022},
  publisher={SIAM}
}

@inbook{Furer2014,
  title = {How Fast Can We Multiply Large Integers on an Actual Computer?},
  ISBN = {9783642544231},
  ISSN = {1611-3349},
  url = {http://dx.doi.org/10.1007/978-3-642-54423-1_57},
  DOI = {10.1007/978-3-642-54423-1_57},
  booktitle = {LATIN 2014: Theoretical Informatics},
  publisher = {Springer Berlin Heidelberg},
  author = {F{\"u}rer,  Martin},
  year = {2014},
  pages = {660–670}
}

@InProceedings{bringmann_2023,
author =	{Bringmann, Karl and Cassis, Alejandro},
title =	{Faster 0-1-Knapsack via Near-Convex Min-Plus-Convolution},
booktitle =	{31st Annual European Symposium on Algorithms (ESA 2023)},
pages =	{24:1--24:16},
series =	{Leibniz International Proceedings in Informatics (LIPIcs)},
ISBN =	{978-3-95977-295-2},
ISSN =	{1868-8969},
year =	{2023},
volume =	{274},
editor =	{G{\o}rtz, Inge Li and Farach-Colton, Martin and Puglisi, Simon J. and Herman, Grzegorz},
publisher =	{Schloss Dagstuhl -- Leibniz-Zentrum f{\"u}r Informatik},
address =	{Dagstuhl, Germany},
URL =		{https://drops.dagstuhl.de/entities/document/10.4230/LIPIcs.ESA.2023.24},
URN =		{urn:nbn:de:0030-drops-186776},
doi =		{10.4230/LIPIcs.ESA.2023.24}
}

@InProceedings{kunnemann_et_al:LIPIcs.ICALP.2017.21,
author =	{K{\"u}nnemann, Marvin and Paturi, Ramamohan and Schneider, Stefan},
title =	{On the Fine-Grained Complexity of One-Dimensional Dynamic Programming},
booktitle =	{44th International Colloquium on Automata, Languages, and Programming (ICALP 2017)},
pages =	{21:1--21:15},
series =	{Leibniz International Proceedings in Informatics (LIPIcs)},
ISBN =	{978-3-95977-041-5},
ISSN =	{1868-8969},
year =	{2017},
volume =	{80},
editor =	{Chatzigiannakis, Ioannis and Indyk, Piotr and Kuhn, Fabian and Muscholl, Anca},
publisher =	{Schloss Dagstuhl -- Leibniz-Zentrum f{\"u}r Informatik},
address =	{Dagstuhl, Germany},
URL =		{https://drops.dagstuhl.de/entities/document/10.4230/LIPIcs.ICALP.2017.21},
URN =		{urn:nbn:de:0030-drops-74688},
doi =		{10.4230/LIPIcs.ICALP.2017.21}
}

@InProceedings{cygan_et_al:LIPIcs.ICALP.2017.22,
author =	{Cygan, Marek and Mucha, Marcin and Wegrzycki, Karol and Wlodarczyk, Michal},
title =	{On Problems Equivalent to {$(\min,+)$}-Convolution},
booktitle =	{44th International Colloquium on Automata, Languages, and Programming (ICALP 2017)},
pages =	{22:1--22:15},
series =	{Leibniz International Proceedings in Informatics (LIPIcs)},
ISBN =	{978-3-95977-041-5},
ISSN =	{1868-8969},
year =	{2017},
volume =	{80},
editor =	{Chatzigiannakis, Ioannis and Indyk, Piotr and Kuhn, Fabian and Muscholl, Anca},
publisher =	{Schloss Dagstuhl -- Leibniz-Zentrum f{\"u}r Informatik},
address =	{Dagstuhl, Germany},
URL =		{https://drops.dagstuhl.de/entities/document/10.4230/LIPIcs.ICALP.2017.22},
URN =		{urn:nbn:de:0030-drops-74216},
doi =		{10.4230/LIPIcs.ICALP.2017.22}
}

@article{bringmann_2023_arxiv,
  doi = {10.48550/ARXIV.2305.01593},
  url = {https://arxiv.org/abs/2305.01593},
  author = {Bringmann,  Karl and Cassis,  Alejandro},
  title = {Faster 0-1-Knapsack via Near-Convex Min-Plus-Convolution},
  publisher = {arXiv},
  journal={arXiv preprint 2305.01593v1},
  year = {2023},
  copyright = {Creative Commons Attribution 4.0 International}
}

@TechReport{stablemult,
  author = 	 {van der Hoeven, Joris},
  title = 	 {Making fast multiplication of polynomials numerically stable},
  institution =  {Universit\'e Paris-Sud},
  year = 	 {2008},
  number = 	 {2008-02},
  address = 	 {Orsay, France}
}

@book{Brent_Zimmermann_2010,
place={Cambridge},
series={Cambridge Monographs on Applied and Computational Mathematics},
title={Modern Computer Arithmetic},
publisher={Cambridge University Press},
author={Brent,
Richard P. and Zimmermann,
Paul},
year={2010},
collection={Cambridge Monographs on Applied and Computational Mathematics}}

\end{document}